\documentclass[twoside,11pt]{article}

\usepackage{blindtext}
\usepackage{amsfonts}
\usepackage{amsmath}
\usepackage{array}
\usepackage{algorithm}
\usepackage{algorithmic}
\usepackage{makecell}
\usepackage{multirow}

\newtheorem{claim}{Claim}
\usepackage[abbrvbib, preprint]{jmlr2e}

\usepackage{color}

\usepackage{lastpage}
\jmlrheading{23}{2022}{1-\pageref{LastPage}}{1/21; Revised 5/22}{9/22}{21-0000}{Shengminjie Chen, Donglei Du, Wenguo Yang, Dachuan Xu, Suixiang Gao}

\ShortHeadings{Continuous Non-monotone DR-submodular Maximization}{Continuous Non-monotone DR-submodular Maximization}
\firstpageno{1}

\begin{document}

\title{Continuous Non-monotone DR-submodular Maximization with Down-closed Convex Constraint}

\author{\name Shengminjie Chen \email chenshengminjie19@mails.ucas.ac.cn \\
       \addr School of Mathematical Sciences\\
       University of Chinese Academy of Science\\
       Beijing, 100049, P.R. China
       \AND
       \name Donglei Du \email ddu@unb.ca \\
       \addr Faculty of Management\\
       University of New Brunswick\\
       Fredericton, E3B 5A3, New Brunswick, Canada
       \AND Wenguo Yang \thanks{Corresponding Author}\email yangwg@ucas.ac.cn \\
       \addr School of Mathematical Sciences\\
       University of Chinese Academy of Science\\
       Beijing, 100049, P.R. China
       \AND
       Dachuan Xu \email xudc@bjut.edu.cn \\
       \addr Beijing Institute for Scientific and Engineering Computing\\
       Beijing University of Technology\\
       Beijing, 100124, P.R. China
       \AND
       Suixiang Gao \email sxgao@ucas.ac.cn \\
       \addr School of Mathematical Sciences\\
       University of Chinese Academy of Science\\
       Beijing, 100049, P.R. China
       }

\editor{}

\maketitle

\begin{abstract}
We investigate the continuous non-monotone DR-submodular maximization problem subject to a down-closed convex solvable constraint. Our first contribution is to construct an example to demonstrate that (first-order) stationary points can have arbitrarily bad approximation ratios, and they are usually on the boundary of the feasible domain. These findings are in contrast with the monotone case where any stationary point yields a $1/2$-approximation (\cite{Hassani}). Moreover, this example offers insights on how to design improved algorithms by avoiding bad stationary points, such as the restricted continuous local search algorithm (\cite{Chekuri}) and the aided measured continuous greedy (\cite{Buchbinder2}). However, the analyses in the last two algorithms only work for the discrete domain because both need to invoke the inequality that the multilinear extension of any submodular set function is bounded from below by its Lovasz extension. Our second contribution, therefore, is to remove this restriction and show that both algorithms can be extended to the continuous domain while retaining the same approximation ratios, and hence offering improved approximation ratios over those in \cite{Bian1} for the same problem. At last, we also include numerical experiments to demonstrate our algorithms on problems arising from machine learning  and artificial intelligence.
\end{abstract}


\tableofcontents

\section{Introduction}\label{section 1}
Submodular optimization is an important topic in discrete optimization. A set function $f : 2^V \rightarrow \mathbb{R}_+$ is submodular if for every $X,Y\subseteq V$, $f(X)+f(Y) \ge f(X \cap Y) + f(X\cup Y)$. Equivalently, submodularity can be expressed by the diminishing return property in set functions, that is, for every $X \subseteq Y \subseteq V$ and for every $u \in V \setminus Y$, $f(\{u\}\cup X) - f(X) \ge f(\{u\} \cup Y) - f(Y)$. 

Submodular functions arise naturally in numerous theoretical studies and practical applications. Submodularity captures one important commonality of many functions investigated in the literature, such as entropy functions, cut functions, coverage functions, etc. Although it can be minimized in polynomial time (\cite{Iwata}), maximizing submodular set functions is an NP-hard problem and can be approximately solved in polynomial time. Seminal works were proposed by \cite{Nemhauser} ($1-e^{-1}$, monotone, cardinality constraints) and \cite{Fisher} ($(1+p)^{-1}$, monotone, intersection of $p$ matroids). Following these work, further results also appeared, such as \cite{Sviridenko} ($1-e^{-1}$, monotone, knapsack constraints) and \cite{Buchbinder} ($1/2$, non-monotone, unconstrained). The above approximation algorithms have been adopted in some application fields such as influence maximization (\cite{Kempe}), active learning (\cite{Golovin}), sensor placement (\cite{Krause}), submodular point processes (\cite{Iyer}).

The multilinear extension has played an important role in maximizing submodular set functions. For example, the continuous greedy ($1-e^{-1}$, monotone, down-closed) (\cite{Calinescu}), the restricted local search ($0.309$, non-monotone, down-closed) (\cite{Chekuri}), the repeated local search ($1/4$, non-monotone, down-closed) (\cite{Chekuri}), the random local search ($2/5$, non-monotone, down-closed) (\cite{Feige}), the measured continuous greedy ($1/e$, non-monotone, down-closed) (\cite{Feldman}), the measured continuous greedy + the continuous double greedy ($0.372$, non-monotone, down-closed) (\cite{Ene}), and the local search + the measured continuous greedy ($0.385$, non-monotone, down-closed) (\cite{Buchbinder2}). These algorithms return a fractional solution, which then can be rounded to an integral solution via such techniques as pipage rounding (\cite{Calinescu}), swap rounding (\cite{Chekuri1}), and contention resolution scheme (\cite{Chekuri}).

However, the aforementioned results concentrate on combinatorial optimization problems; that is, feasible solutions belong to a discrete domain (binary or integer vectors). Recently, submodular optimization in the continuous domain (namely real vectors) becomes an attractive research field due to its vast applications, in particular, machine learning and artificial intelligence. \cite{Bach} extended the submodular set function minimization problem to the continuous domain. \cite{marinacci2005ultramodular} and \cite{Bian} characterize the difference between submodular functions and DR-submodular functions in the continuous domain. \cite{Bian} proposed a $(1-e^{-1})$ approximation algorithm to maximize monotone continuous DR-submodular functions under down-closed convex constraints. \cite{Hassani} adopted the projected gradient ascent method to maximize monotone continuous DR-submodular functions over general convex constraints with a $1/2$ approximation guarantee. For non-monotone cases, \cite{Bian1} adopted a two-phase Frank-Wolfe algorithm with a $1/4$ approximation ratio and extended the measured continuous greedy from the multilinear extension of submodular set functions to continuous DR-submodular functions subject to down-closed convex constraints with a $1/e$ approximation ratio. \cite{Niazadeh} proposed a $1/2$ approximation algorithm subject to box constraints, which improves prior works in \cite{Bian}. If the feasible domain is a general convex set containing the origin, \cite{du2022improved} proposed a Frank-Wolfe algorithm with a $1/4$ approximation ratio, which is proved to be tight by \cite{Mualem}. In addition, \cite{Du1} proposed a systematic framework, namely the Lyapunov approach, for approximation algorithms design and analyses in the continuous domain and then a systematic approach to discretize the algorithm. These results have been applied to many application problems, such as revenue maximization (\cite{Bian3}), mean-field inference (\cite{Bian2}), robust budget allocation (\cite{Staib}), determinantal point processes (\cite{Gillenwater}), etc.

Although many advances have been made for DR-submodular maximization in both discrete and continuous domains, there are still many important open problems to be addressed.

\begin{itemize}
    \item Is the approximation ratio 1/2 in \cite{Hassani}  still valid for the stationary points in the non-monotone case? If not, what is the approximation ratio in the non-monotone case?
    \item The restricted continuous local search in \cite{Chekuri} and the aided measured continuous greedy in \cite{Buchbinder2} only work for the discrete domain because their proofs strongly depend on the multilinear extension and the Lovasz extension. Can we extend these algorithms to the continuous domain without worsening the approximation ratios? if yes, can the continuous domain offers us more flexibility in tuning the parameters used in the algorithms, and hence potentially improving the approximation ratios?
\end{itemize}

The main focus of this work is to answer these questions on the following problem.
\begin{equation}\label{eq:main_prob}
    \max_{\mathbf{x} \in \mathcal{P}} F(\mathbf{x})
\end{equation}
where $\mathcal{P}\subseteq [0,1]^n$ is down-closed, convex, and solvable, and $F(\mathbf{x}): [0,1]^n\mapsto \mathbb{R}_{\ge 0}$ is continuous, non-negative, non-monotone, and DR-submodular.

{\bf Contribution.} To answer the first question, we construct an example from the multilinear extension of regular coverage set functions, which is non-negative, non-monotone and DR-submodular. With this concrete example, we propose a minimization problem to show that the ratio between the values of the worst stationary point and the optimal solution approaches $0$, implying that the stationary point can have arbitrarily bad approximation ratio for non-monotone DR-submodular objective functions. Moreover, the worst stationary point tends to be close to the boundary of the feasible domain. These fundings may explain why the restricted continuous local search (\cite{Chekuri}) and the aided measured continuous greedy (\cite{Buchbinder2}) produce better approximation ratios. 

For the second question, we offer an essentially different analysis to show the same upper bounds on the optimal value used in the $0.309$ (\cite{Chekuri}) and the $0.385$  (\cite{Buchbinder2}) approximation algorithm for the continuous domain. The key new ingredient in our analysis is to eliminate the multilinear and the Lovasz extensions and utilize solely the DR-submodularity. The aforementioned second upper bound, together with a tailor-made  Lyapunov function, allow us to adopt the Lyapunov framework introduced in \cite{Du-submodular-book2022} to present an $0.385$ approximation algorithm for the continuous domain problem in (\ref{eq:main_prob}), and hence extending the \emph{aided measured continuous greedy algorithm} from the discrete domain to the continuous domain. Moreover, we show that this $0.385$ approximation algorithm is `optimal' in the sense that the algorithm is the optimal solution of certain variation problem. This perspective offers insights on which directions to devote (or avoid) our efforts to potentially design improved approximation algorithms for the problem.

At last, we also include numerical experiments to demonstrate our algorithms on problems arising from machine learning  and artificial intelligence.
\begin{table}
    \centering
    \caption{Summary of DR-submodular maximization with down-closed convex constraints}
    \begin{tabular}{m{6.5cm} m{5cm} m{3cm}}
        \makecell[c]{Reference} & \makecell[c]{Setting} & \makecell[c]{Guarantee} \\
         \hline
         \makecell[c]{Repeated Local Search\\ (\cite{Chekuri})}& \makecell[c]{discrete, non-monotone\\ DR-submodular}  &\makecell[c]{$1/4 - O(\varepsilon)$} \\ 
         \hline
         \makecell[c]{Two Stage Frank-Wolfe \\ (\cite{Bian1})} & \makecell[c]{continuous, non-monotone\\ DR-submodular} & \makecell[c]{$1/4 - O(\varepsilon)$} \\
         \hline 
         \makecell[c]{Continuous Greedy \\ (\cite{Calinescu})}& \makecell[c]{discrete, monotone\\ DR-submodular}  & \makecell[c]{$1-e^{-1} - O(\varepsilon)$}  \\
         \hline
         \makecell[c]{Frank-Wolfe Variant\\ (\cite{Bian})}& \makecell[c]{continuous, monotone\\ DR-submodular}  & \makecell[c]{$1-e^{-1} - O(\varepsilon)$} \\
 		\hline
         \makecell[c]{Measured Continuous Greedy \\ (\cite{Feldman})} & \makecell[c]{discrete, non-monotone\\ DR-submodular}  & \makecell[c]{$e^{-1} - O(\varepsilon)$}\\
         \hline 
         \makecell[c]{Frank-Wolfe Variant\\ (\cite{Bian1})} & \makecell[c]{continuous, non-monotone\\ DR-submodular}  & \makecell[c]{$e^{-1} - O(\varepsilon)$}\\
         \hline
         \makecell[c]{Restricted Local Search \\ (\cite{Chekuri})} & \makecell[c]{discrete, non-monotone\\ DR-submodular} & \makecell[c]{$0.309 - O(\varepsilon)$} \\
         \hline 
         \makecell[c]{Aided Measured Continuous Greedy \\ (\cite{Buchbinder2})} & \makecell[c]{discrete, non-monotone\\ DR-submodular} & \makecell[c]{$0.385-O(\varepsilon)$} \\ 
         \hline 
         \makecell[c]{Stationary Points \\ \cite{Hassani}}& \makecell[c]{monotone \\ DR-submodular}  &\makecell[c]{$1/2$} \\ 
         \hline
         \makecell[c]{Stationary Points \\ \textbf{This work}}& \makecell[c]{non-monotone \\ DR-submodular}  &\makecell[c]{$0$} \\ 
         \hline
         \makecell[c]{Stationary Points \\ (reduced feasible domain)\\ \textbf{This work} } & \makecell[c]{non-monotone \\ DR-submodular}  & \makecell[c]{$0.309 - O(\varepsilon)$} \\
         \hline 
         \makecell[c]{Aided Frank-Wolfe Variant \\ \textbf{This work}} & \makecell[c]{continuous, non-monotone\\ DR-submodular} & \makecell[c]{$0.385 - O(\varepsilon)$ } \\
         \hline
    \end{tabular}
    \label{summary results}
\end{table}

{\bf Applications.} Continuous DR-submodular functions are ubiquitous in many applications. Below are some examples that are common in discrete optimization, machine learning, and data mining. 
\begin{itemize}
    \item Multilinear extensions of submodular set functions (ME): The multilinear extension of each submodular set function is a special DR-submodular function that are infinitely differentiable. The multilinear extension of the set function $f : 2^{[n]} \rightarrow \mathbb{R}$ is defined as the function $F(\mathbf{x}): [0,1]^{n}\mapsto \mathbb{R}$ with 
    \[
    F(\mathbf{x}) = \sum_{S \subseteq [n]} f(S) \prod_{i \in S}x_i \prod_{j \in [n]\setminus S} (1-x_j),
    \] 
    where $[n]:=\{1,\ldots, n\}$.
        \item Non-convex/concave quadratic programming (NQP): The goal is to maximize $F(\mathbf{x}) = \frac{1}{2} \mathbf{x}^T \mathbf{H} \mathbf{x} + \mathbf{h}^T \mathbf{x} +c$ subject to linear constraints, where all entries of Hessian $H$ are non-positive. This type of problems arises in many applications such as scheduling (\cite{Skutella}), and demand forecasting (\cite{Ito}), among many others. 
    
    \item Softmax extension of determinantal point processes (DPPs): DPPs is an important probabilistic model in machine learning, which can be adopted as generative models in applications such as document summarization (\cite{Gillenwater,Bian1}), news threading tasks (\cite{Kulesza}), etc. The formal description is as follows:
    \begin{equation}\label{dpps}
        F(\mathbf{x}) = \log \det \left[\text{diag} (\mathbf{x}) (\mathbf{L}-\mathbf{I}) + \mathbf{I} \right]
    \end{equation}
    where $\text{diag}(\mathbf{x})$ is the diagonal matrix with the vector $\mathbf{x} \in [0,1]^n$ and $\mathbf{I}$ is the identity matrix.
    
    \item Revenue/Profit maximization in online social networks  (RM/PM): This problem comes from the ``word-of-mouth" effects in online social networks. Sellers provide a subset of users in online social networks with some free samples, such as discounts, free products, etc, to maximize the total revenue/profit when other users buy their product. In \cite{Bian,Tang}, the objective function has two parts, the revenue gain and the revenue loss/cost. In this work, we adopt the form in \cite{Bian}, where $x_t$ represents the free sample of user $t$, $R_s(\mathbf{x})$ is the non-negative, non-decreasing, DR-submodular revenue gain, $\bar{R}_t(\mathbf{x})$ is the non-positive, non-increasing DR-submodular revenue loss, and $\phi(x_t)$ is the modular revenue gain from user $t$ who received the free sample.
    \begin{equation}\label{revenue function}
        F(\mathbf{x}) = \alpha \sum_{s:x_s = 0} R_s(\mathbf{x}) + \beta \sum_{t:x_t \neq 0} \phi (x_t) + \gamma \sum_{t:x_t \neq 0} \bar{R}_t(\mathbf{x})
    \end{equation}
\end{itemize}

{\bf Organization.} The rest of this work is organized as follows. First, we present some basic definitions and lemmas in Section 2. In Section 3, we illustrate the negative result of stationary points for non-monotone scenarios by an example and prove the $0.309$ approximation guarantee of stationary points in the reduced feasible domain. In Section 4, we propose unified and novel analyses of the upper bound on the optimal value for the continuous domain. Then we use a tailor-make Lyapunov function to our problem to present a $0.385$ approximation algorithm, which is shown to be the optimal solution of a variation problem. Numerical experiments are presented in Section 5. We provide concluding remarks in Section 6.

\section{Preliminary}\label{section 2}
We first introduce some relevant definitions and lemmas. The first concept is the submodular function on the real lattice $\mathbb{R}^n$, which is a special case of the submodular function on arbitrary lattice (see e.g.,~\cite{Topkis1998}). 
\begin{definition}
    A function $F(\mathbf{x}): \mathbb{R}^n\mapsto \mathbb{R}$ is submodular if 
    \begin{equation}
        F(\mathbf{x}) +F(\mathbf{y}) \ge F(\mathbf{x} \vee \mathbf{y}) + F(\mathbf{x} \wedge \mathbf{y}), \forall \mathbf{x},\mathbf{y} \in \mathbb{R}^n,
    \end{equation}
    where $\vee$ and $\wedge$ are the coordinate-wise maximum and minimum: ${x}_i \vee {y}_i = \max \{x_i,y_i\}$ and ${x}_i \wedge {y}_i= \min \{x_i,y_i\}$, $\forall i \in [n]$. 
\end{definition}
Note that a twice-differentiable function $F(\mathbf{x})$ is submodular iff
    \begin{equation}
     \cfrac{\partial ^2 F(\mathbf{x})}{\partial x_i \partial x_j} \le 0, \text{ for }\forall \mathbf{x} \in \mathcal{X}, \forall i \neq j.
        \end{equation}

The above definition implies that submodular functions are generally non-convex and non-concave. In addition, the above definition has no restriction on the diagonal elements of the Hessian matrix, i.e., $\frac{\partial^2 F(\mathbf{x})}{\partial x^2_i} \in \mathbb{R}$. Naturally, if $\frac{\partial^2 F(\mathbf{x})}{\partial x^2_i}$ is further required to be non-positive, then it implies a stronger submodularity, namely DR-submodularity, which is coined in \cite{Bian}. However, the same concept (under a different name) was already investigated systematically in \cite{marinacci2005ultramodular} in more general settings.

\begin{definition}[\cite{marinacci2005ultramodular,Bian}]
    A function $F(\mathbf{x})$ is defined on subsets of $\mathbb{R}^n: \mathcal{X}  = \prod_{i =1} ^n \mathcal{X}_i$, where each $\mathcal{X}_i$ is a compact subset of $\mathbb{R}$. If $F(\mathbf{x}) : \mathcal{X} \rightarrow \mathbb{R}$ is DR-submodular, for $\forall \mathbf{x} \le \mathbf{y} \in \mathcal{X}$, $\forall i \in [n]$, $\forall k \in \mathbb{R}_+$ such that $k\mathbf{e}_i + \mathbf{x}$ and $k \mathbf{e}_i + \mathbf{y}$ are still in $\mathcal{X}$, the following inequality holds
    \begin{equation}
        F(\mathbf{x} + k \mathbf{e}_i) - F(\mathbf{x}) \ge F(\mathbf{y} + k \mathbf{e}_i) - F(\mathbf{y})
    \end{equation}
    where $\mathbf{e}_i$ represents the unit vector and $\mathbf{x}\le \mathbf{y}$ means $x_i \le y_i, \forall i \in [n]$. A twice-differentiable function $F(\mathbf{x})$ is DR-submodular iff
    \begin{equation}
     \cfrac{\partial ^2 F(\mathbf{x})}{\partial x_i \partial x_j} \le 0, \text{ for }\forall \mathbf{x} \in \mathcal{X}
    \end{equation}
\end{definition}

It is easy to show that the multilinear extension of a submodular set function is infinitely differentiable and DR-submodular. In addition, a differentiable function is DR-submodular iff  $\forall \mathbf{x}\le \mathbf{y} \in \mathcal{X}$, $\nabla F(\mathbf{x}) \ge \nabla F(\mathbf{y})$. Based on this observation, naturally, there is the following lemma.

\begin{lemma}[\cite{marinacci2005ultramodular,Calinescu,Bian}]\label{lemma 3}
    If $F(\mathbf{x})$ is DR-submodular, then $F(\mathbf{x})$ is concave along any non-negative directions. 
\end{lemma}

Lemma \ref{lemma 3} shows that DR-submodular functions possess certain `concavity' to some extent, implying the following result.
\begin{lemma}[A rephrased version of Lemma 3 in \cite{Bian1}]\label{lemma 4}
    If $F(\mathbf{x})$ is DR-submodular and non-negative, then for any $\mathbf{x},\mathbf{y} \in \mathcal{X}$, we have:
    \begin{equation}
        \begin{aligned}
            F(\mathbf{x}\vee \mathbf{y}) & \ge (1-\alpha) F(\mathbf{y}) \\
            F(\mathbf{x} \wedge \mathbf{y}) & \ge \beta F(\mathbf{y})
        \end{aligned}
    \end{equation} 
    where $\alpha \ge \max_{i \in [n]} \frac{x_i-l_i}{u_i-l_i}$, $\beta \le \min_{i \in [n]} \frac{x_i -l_i}{u_i -l_i}$, and $l_i,u_i$ are the lower bound and the upper bound of $\mathcal{X}_i$, respectively.
\end{lemma}
\begin{lemma}\label{lemma 5}
    For a DR-submodular function $F(\mathbf{x})$ defined on $\mathcal{X} = [\mathbf{l},\mathbf{u}]$, where $\mathbf{l}=(l_1,l_2,...,l_n)$, $\mathbf{u}=(u_1,u_2,...,u_n)$ are the lower bound and the upper bound of $\mathcal{X}$, we denote $\alpha_i = \frac{x_i - l_i}{u_i-l_i}$ for all $i \in [n]$ and $\mathbf{l}_S=(l_{s_1},l_{s_2},...,l_{s_n})$ by $l_{s_i}=l_i$ if $i\in S$ otherwise $l_{s_i} = 0$, similarly, $\mathbf{u}_{[n]\setminus S}$. Then, we have:
    \begin{equation}\label{concave}
        F(\mathbf{x}) \ge \sum_{S \subseteq [n]} F(\mathbf{l}_{[n]\setminus S} \vee \mathbf{u}_{S}) \prod_{i \in S} \alpha_i \prod_{j \in [n] \setminus S } (1-\alpha_j)
    \end{equation}
\end{lemma}
\begin{proof}
    If $n=1$, $F(\mathbf{x})= F(\alpha_1 u_1  + (1-\alpha_1) l_1) \ge \alpha_1 F(u_1) + (1-\alpha_1) F(l_1)$. By induction, we assume that Eq.(\ref{concave}) holds on $n=k-1$. Therefore, when $n=k$, we have: 
    \begin{equation*}
        \begin{aligned}
            F(\mathbf{x}) & = F(\alpha_k(u_k \mathbf{e}_k \vee \mathbf{x}_{[k-1]}) + (1- \alpha_k) (l_k \mathbf{e}_k \vee \mathbf{x}_{[k-1]}) ) \\
            &  \ge \alpha_k F(u_k \mathbf{e}_k \vee \mathbf{x}_{[k-1]}) + (1-\alpha_k) F(l_k \mathbf{e}_k \vee \mathbf{x}_{[k-1]}) \\
            & \ge \alpha_k \sum_{S \subseteq [k-1]} F(\mathbf{l}_{[k-1]\setminus S } \vee \mathbf{u}_{S \cup \{k\} }) \prod_{i \in S} \alpha_i \prod_{j \in [k-1] \setminus S } (1-\alpha_j) \\
            & \quad + (1-\alpha_k) \sum_{S \subseteq [k-1]} F(\mathbf{l}_{[k]\setminus S} \vee \mathbf{u}_{S}) \prod_{i \in S} \alpha_i \prod_{j \in [k-1] \setminus S } (1-\alpha_j)\\
            & = \sum_{S \subseteq [k]} F(\mathbf{l}_S \vee \mathbf{u}_{[k]\setminus S}) \prod_{i \in S} \alpha_i \prod_{j \in [k] \setminus S } (1-\alpha_j)
        \end{aligned}
    \end{equation*}
\end{proof}

We need the concept of stationary point for problem (\ref{eq:main_prob}), where we only need to assume that $F$ is differentiable.
\begin{definition}\label{definition 6}
   A feasible solution $\mathbf{x} \in \mathcal{P}$ for problem (\ref{eq:main_prob}) is a stationary point, when $F$ is differentiable, if it satisfies the first-order necessary condition, 
    \begin{equation}
        \max_{\mathbf{y} \in \mathcal{P}} \langle \nabla F(\mathbf{x}), \mathbf{y}-\mathbf{x} \rangle \le 0
    \end{equation}
\end{definition}

When $F$ is further assumed to be DR-submodular, as assumed in problem (\ref{eq:main_prob}), we have the following inequality.
\begin{lemma}[A rephrased version of Corollary 3.3 in \cite{Chekuri}]\label{lemma 7}
    If $\mathbf{x}$ is a stationary point for problem (\ref{eq:main_prob}), then 
    \begin{equation}\label{eq:stat_ineq}
        2 F(\mathbf{x}) \ge F(\mathbf{x} \vee \mathbf{y}) + F(\mathbf{x} \wedge \mathbf{y}),  \forall \mathbf{y} \in \mathcal{X} 
    \end{equation}
\end{lemma}
Lemma~\ref{lemma 7} immediately implies the following result by a simple convex combination of (\ref{eq:stat_ineq}) with another inequality resulting from choosing $\mathbf{y}:=\mathbf{x}\wedge \mathbf{y}\in \mathcal{P}$ due to down-closeness in (\ref{eq:stat_ineq}).
\begin{lemma}\label{lemma 8}
     If $\mathbf{x}$ is a stationary point for (\ref{eq:main_prob}), then we have
    \begin{equation}
        F(\mathbf{x}) \ge \gamma F(\mathbf{x} \vee \mathbf{y}) + (1-\gamma) F(\mathbf{x} \wedge \mathbf{y}), \forall \mathbf{y} \in \mathcal{X}, \gamma \in [0,0.5] 
    \end{equation}
\end{lemma}

\section{Performance of Stationary Points} \label{sec:stationary point}
Although \cite{Hassani} show that the ratio between the values of the worst stationary point and the optimal solution for monotone DR-submodular objective functions is $1/2$, it is still an open problem whether the same approximation ratio still holds for the non-monotone case. We answer this question negatively by showing that stationary points can have arbitrarily bad approximation ratios. For this purpose, we construct an example to show that the value ratio between the worst stationary point and the optimal solution approaches $0$ for the non-monotone case. In addition, we observe an interesting phenomenon that the worst stationary point tends to be close to the boundary of the feasible domain, implying that removing stationary points near the boundary of the feasible domain can potentially improves the approximation ratio of stationary points, as have been done previously in the literature, such as \cite{Chekuri,Buchbinder2}. Our example may explain their success in obtaining improved approximation algorithms. 

The rest of this section provides a new analysis, which removes the reliance on the multilinear and the Lovasz extensions used in \cite{Chekuri}, to prove that their \emph{restricted local search} algorithm is also a $0.309$ approximation solution in the continuous domain. 

\subsection{Negative Result for Stationary Points}\label{subsection 3.1}
The conventional coverage function is one of the famous submodular set functions, which is monotone non-decreasing, and non-negative. Inspired by the coverage function, we construct a regular coverage function which is the difference between the coverage function and the size of the selected solution. The formal description is as follows.

\begin{definition}\label{definition 9}
    For a ground set $V = \{1,2,...,n\}$ and a subset collection $\{S_1, ... ,S_n\}$ where $S_i \subseteq V$, the regular coverage function is defined as $f(X) =  |\bigcup_{i \in X} S_i | - |X|$. 
\end{definition}

According to the above definition, although the regular coverage function is not always non-negative, we can construct a concrete example of the regular coverage function such that it is non-monotone and non-negative.

\begin{example}\label{example 1}
    Let $V=\{1,2,...,2k+1\}$, $S_i = \{i,2k+1\}$ for $i \in \{1,2,...,k\}$, $S_i = \{i\}$ for $i \in \{k+1,k+2,...,2k\}$, and $S_{2k+1} = \{1,2,...,k, 2k+1 \}$. It is not hard to show that the regular coverage function for this instance is non-monotone, non-negative, and submodular. 
\end{example}

\begin{example}
    As an illustration of the above example, choosing $k=1$ gives us $V=\{1,2,3\}$, $S_1=\{1,3\}$, $S_2 = \{2\}$, $S_3 = \{1,3\}$. Then, $f(\{1\}) = 1$, $f(\{2\}) = 0$, $f(\{3\})= 1$, $f(\{1,2\})=1$, $f(\{1,3\})=0$, $f(\{2,3\})=1$, $f(\{1,2,3\}) = 0$.
\end{example}

Based on the setup of Example \ref{example 1}, let $F:[0,1]^{2k+1} \rightarrow \mathbb{R}_{\ge 0}$ be the multilinear extension of the regular coverage function $f:2^V \rightarrow \mathbb{R}_{\ge 0}$. $F(\mathbf{x})$ and $\nabla F(\mathbf{x})$ have the following closed-form:

\begin{equation}
\begin{array}{c}
    \begin{aligned}
        F(\mathbf{x}) & = \sum_{S \subseteq V} f(S) \cdot \prod_{i \in S} x_i \cdot \prod_{i \notin S} (1-x_i) = \sum_{S \subseteq V} (|\cup_{i \in S} S_i | - |S| )\cdot \prod_{i \in S} x_i \cdot \prod_{i \notin S} (1-x_i) \\
        & = \sum_{S \subseteq V} |\cup_{i \in S} S_i | \cdot \prod_{i \in S} x_i \cdot \prod_{i \notin S} (1-x_i) - \sum_{S \subseteq V} |S| \cdot \prod_{i \in S} x_i \cdot \prod_{i \notin S} (1-x_i) \\
        & = k+1-\left(1-x_{2 k+1}\right) \prod_{i=1}^k\left(1-x_i\right)-\left(1-x_{2 k+1}\right)\left(k-\sum_{i=1}^k x_i\right)+\sum_{i=k+1}^{2 k} x_i - \sum_{i=1}^{2 k +1} x_i \\
        & = k+1-\left(1-x_{2 k+1}\right) \prod_{i=1}^k\left(1-x_i\right)-\left(1-x_{2 k+1}\right)\left(k-\sum_{i=1}^k x_i\right) - \sum_{i=1}^{k} x_i -x_{2k+1}
    \end{aligned}\\
    \cfrac{\partial F(\mathbf{x})}{\partial x_i} = \begin{cases} (1-x_{2k+1}) \cdot \prod_{j=1, j\neq i}^{k}(1-x_j) -x_{2k+1} & \text { if } i \le k\\ 0 & \text { if } k+1 \le i \le 2k\\ \prod_{j=1}^{k}(1-x_j) + k-\sum_{i=1}^k x_i - 1 & \text { if } i =2k+1\end{cases}
\end{array}
\end{equation}

Our goal is to calculate the ratio between the worst stationary point and the optimal solution. In other words, we want to minimize the ratio between stationary points in the feasible domain and the optimal solution, i.e., minimize the ratio $\frac{F(\mathbf{x})}{F(\mathbf{x}^*)}$ where $\mathbf{x}$ is an arbitrary stationary point in the feasible domain and $\mathbf{x}^*$ is the optimal solution. To achieve our objective, we construct the following minimization problem.

\begin{equation}\label{minimizing problem}
    \begin{aligned}
         \min & \quad \cfrac{F(\mathbf{x})}{F(\mathbf{y})}\\
        \text{s.t.}  \quad & \mathbf{x} \text{ is a stationary point} \\
        & \mathbf{x},\mathbf{y} \text{ are feasible} 
    \end{aligned}
\end{equation}

Note that the objective function in Eq.(\ref{minimizing problem}) is to minimize the ratio $\frac{F(\mathbf{x})}{F(\mathbf{y})}$ and $\mathbf{y}$ is in the denominator. It means that the ratio $\frac{F(\mathbf{x})}{F(\mathbf{y})}$ reaches the global minimum if and only if $\mathbf{y}$ is the optimal solution of the function $F$ in the feasible domain. It implies that constraints on the above minimization problem are that $\mathbf{x}$ is a stationary point and $\mathbf{x},\mathbf{y}$ are in the feasible domain. Because the multilinear extension $F(\mathbf{x})$ is defined in $[0,1]^{2k+1}$ which is a special down-closed convex constraint, an immediate observation is that $\mathbf{z}^* = \arg\max_{\mathbf{z} \in [0,1]^{2k+1}} \left\langle\nabla F(\mathbf{x}), \mathbf{z}-\mathbf{x}\right\rangle = \mathbf{1}_{\nabla F(\mathbf{x})_{\ge 0}}$, where the last term is the $n$-dimensional binary vector whose $i$-th coordinate ($ i\in [n]$) is 1 when $\frac{\partial F(\mathbf{x})}{\partial x_i} \ge 0$ and 0 otherwise. 
We can therefore simplify the constraint as follows: 
\begin{equation*}
    \max_{\mathbf{z} \in [0,1]^{2k+1}} \langle\nabla F(\mathbf{x}) , \mathbf{z}-\mathbf{x}\rangle \le 0 
     \iff \langle \nabla F(\mathbf{x}) , \mathbf{1}_{\nabla F(\mathbf{x})_{\ge 0}}-\mathbf{x}\rangle \le 0
\end{equation*}
Now, we have the following optimization problem to reveal the ratio between the worst stationary point and the optimal solution.
\begin{equation}\label{ratio problem}
    \begin{aligned}
         \min & \qquad \cfrac{F(\mathbf{x})}{F(\mathbf{y})}\\
        \text{s.t.}  &\left \langle \nabla F(\mathbf{x}) , \mathbf{1}_{\nabla F(\mathbf{x})_{\ge 0}}-\mathbf{x}\right\rangle \le 0 \\
        & \mathbf{x},\mathbf{y} \in [0,1]^{2k+1}
    \end{aligned}
\end{equation}

Specially, we choose two special points $\mathbf{x}$ and $\mathbf{y}$, where $x_i = 1$, $y_i = 0$ for $i= 1,...,2k$, and $x_{2k+1}=0$, $y_{2k+1} = 1$. In addition, we have $\frac{\partial F(\mathbf{x})}{\partial x_i} = 0$ for $i=1,...,2k$, $\frac{\partial F(\mathbf{x})}{\partial x_{2k+1}} = -1$, $\frac{\partial F(\mathbf{y})}{\partial y_i} = -1$ for $i=1,...,k$, $\frac{\partial F(\mathbf{y})}{\partial y_i} = 0$ for $i=k+1,...,2k$, and $\frac{\partial F(\mathbf{y})}{\partial y_{2k+1}} = k$. Therefore, 
\begin{eqnarray*}
\max_{\mathbf{z} \in [0,1]^{2k+1}} \langle\nabla F(\mathbf{x}) , \mathbf{z}-\mathbf{x}\rangle &=& \langle\nabla F(\mathbf{x}) , \mathbf{1}_{[2k]}-\mathbf{x}\rangle = 0\\
\max_{\mathbf{z} \in [0,1]^{2k+1}} \langle\nabla F(\mathbf{y}) , \mathbf{z}-\mathbf{y}\rangle &=& \langle\nabla F(\mathbf{x}) , \mathbf{1}_{[k+1,2k+1]}-\mathbf{y}\rangle = 0
\end{eqnarray*}
It implies that $\mathbf{x}$ and $\mathbf{y}$ are stationary points, i.e., $\frac{F(\mathbf{x})}{F(\mathbf{y})} = \frac{1}{k}$. In other words, the ratio between the worst stationary point and the optimal solution is not greater than $\frac{1}{k}$, which approaches $0$ with increasing $k$.

In terms of the numerical solution, solving Eq.(\ref{ratio problem}) under $k=2$, $3$, $5$, $10$, $20$, $30$, $50$, the ratio between the worst stationary point and the optimal solution decreases with $k$ and tends to $0$. Details are in Tab.\ref{Ratio}. Although Eq.(\ref{ratio problem}) is non-convex and non-concave, the solution returned by the optimization algorithm is a stationary point of Eq.(\ref{ratio problem}). It represents that the global minimum of Eq.(\ref{ratio problem}) is no more than that of the solution returned by the optimization algorithm. In addition, there is an interesting phenomenon that the maximal component of the bad stationary point tends to be close to $1$. Therefore, we have the following claim.

\begin{claim}\label{claim}
    For non-monotone continuous DR-submodular functions, the ratio $\frac{F(\mathbf{x})}{F(\mathbf{x}^*)}$ is $0$, where $\mathbf{x}$ is the worst stationary point and $\mathbf{x}^*$ is the optimal solution. It implies that the stationary point will be arbitrarily bad in non-monotone cases. Furthermore, the worst stationary point tends to be close to the boundary of the feasible domain. 
\end{claim}

\begin{table}
    \centering
    \caption{Ratios between the worst stationary point and the optimal solution}
    \label{Ratio}
    \begin{tabular}{m{2cm} m{1.2cm} m{1.2cm} m{1.2cm} m{1.2cm} m{1.2cm} m{1.2cm} m{1.2cm}}
        \hline
        & $k=2$ & $k=3$ &$k=5$ & $k=10$ & $k=20$ & $k=30$ & $k=50$ \\
        \hline
        Ratio $\frac{F(\mathbf{x})}{F(\mathbf{y})}$& $0.4142$ & $0.3222$ & $0.1999$ & $0.1000$ & $0.0500$ & $0.0333$ & $ 0.0200$\\
        \hline
        $\max_{i \in [2k+1]} x_i$& $0.5858$ & $0.6778$ & $0.9999$ & $0.9999$ & $0.9804$ & $1.0000$ & $0.9942$\\
        \hline
    \end{tabular}
\end{table}

\subsection{Stationary Points in Reduced Feasible Domain}\label{subsection 3.2}
According to Tab.\ref{Ratio} and Claim \ref{claim}, if we want stationary points to yield better approximation ratio, removing some of the bad stationary points that are near the boundary of the feasible domain is a natural approach. To some extent, it explains why the restricted continuous local search by \cite{Chekuri} has a better approximation ratio than the continuous local search for maximizing the multilinear extension of submodular set functions. Moreover,  the proofs in \cite{Chekuri} heavily depend on the fact that the multilinear extension is no less than the Lovasz extension of any submodular set function. Therefore leaving open the question whether the approximation ratio $0.309$ can also be achieved in the continuous domain for DR-submodular functions that are not arising from the multilinear extension. 

In problem (\ref{eq:main_prob}), $\forall m \in [0,1]$, define the reduced feasible domain $\mathcal{P} \cap [0,m]^n$. We introduce a projection vector $\mathbf{w}$ that depends on the stationary point $\mathbf{x} \in \mathcal{P}\cap[0,m]^n$, such that, for each point $\mathbf{z}\in \mathcal{P}$ but $\mathbf{z}\notin \mathcal{P}\cap [0,m]^n$, there is a vector $\mathbf{z}' = \mathbf{z} \wedge \mathbf{w}\in \mathcal{P}$ but $\mathbf{z}' \notin \mathbf{int}(\mathcal{P}\cap [0,m]^n)$ satisfying the following lemma.


\begin{lemma}[A rephrased version of Lemma 3.16 in \cite{Chekuri}]\label{lemma 10}
Let $\mathbf{x}$ be a stationary point in $\mathcal{P}\cap[0,m]^n$ and $\mathbf{z}$ be any point in $\mathcal{P}$. Define $\mathbf{w} = (w_1,w_2,...,w_n) \in [0,1]^n$ by $w_i = m$ if $x_i \ge m-\varepsilon$ and $w_i =1$ if $x_i < m- \varepsilon$. Then, for $\mathbf{z}' = \mathbf{z} \wedge \mathbf{w}$, we have:
\begin{equation}
    2 F(\mathbf{x}) \ge F(\mathbf{x}\wedge \mathbf{z}') + F(\mathbf{x} \vee \mathbf{z}')
\end{equation}
\end{lemma}


To prove the approximation ratio of $\mathbf{x}$, we next prove the following lower bounds on $F(\mathbf{x} \vee \mathbf{z}')$ and $F(\mathbf{x} \wedge \mathbf{z}')$, respectively, without using the multilinear and Lovasz extensions as have been done in \cite{Chekuri}. Therefore the following lemmas hold not only for DR-submodular functions arising as the multilinear extensions of submodular set functions but also for general DR-submodular functions defined directly on the continuous domain. We propose a simplified and unified analysis by only using the DR-submodularity of $F$. 
\begin{lemma}\label{lemma 11}
    Let $\mathbf{x}\in\mathcal{P}\cap[0,m]^n$ be a stationary point in the reduced feasible domain and $\mathbf{x}^*\in \mathcal{P}$ be the optimal solution. Denote $A=\{i \in [n]: x_i \ge m-\varepsilon\}$ and $\hat{A} = [n]\setminus A$. For $\mathbf{z}' = \mathbf{x}^* \wedge \mathbf{w}$, we have: 
    \begin{equation}
        F(\mathbf{x} \vee \mathbf{z}') \ge (1-m)m F(\mathbf{x}^*) + (1-m)^2 F(\mathbf{x}^*_{\hat{A}})
    \end{equation}
    where $\mathbf{x}^*_{\hat{A}} $ is a vector that agrees with $\mathbf{x}^*$ on coordinates belonging to $\hat{A}$, and equals to 0 on other coordinates, i.e., $\mathbf{x}^*_A \wedge \mathbf{x}^*_{\hat{A}} = \mathbf{0}$, $\mathbf{x}^*_A \vee \mathbf{x}^*_{\hat{A}} = \mathbf{x}^*$.   
\end{lemma}

\begin{proof}
    Firstly, define a function $G$ on $[\mathbf{0}, \mathbf{x}^*_A]$; that is, $\forall \mathbf{y} \in [\mathbf{0},\mathbf{x}^*_A]$, $G(\mathbf{y}) := F(\mathbf{x}^*_{\hat{A}}+\mathbf{y}) - F(\mathbf{x}^*_{\hat{A}})$. Note that $G(\mathbf{0}) = 0$. We now show that $G$  is DR-submodular.   $\forall i \in A$ and $k_i \in \mathbb{R}_+$, for all $\mathbf{a} \le \mathbf{b} \in  [\mathbf{0},\mathbf{x}^*_A]$ and $\mathbf{b} + k_i \mathbf{e}_i \in [\mathbf{0},\mathbf{x}^*_A]$.   
    \begin{equation*}
        \begin{aligned}
            G(k_i \mathbf{e}_i + \mathbf{a}) - G(\mathbf{a}) & = F((k_i \mathbf{e}_i + \mathbf{a}) + \mathbf{x}^*_{\hat{A}}) - F(\mathbf{a} + \mathbf{x}^*_{\hat{A}}) = F(k_i \mathbf{e}_i|\mathbf{a} + \mathbf{x}^*_{\hat{A}} ) \\
            & \ge F(k_i \mathbf{e}_i|\mathbf{b} + \mathbf{x}^*_{\hat{A}} ) = F((k_i \mathbf{e}_i + \mathbf{b}) + \mathbf{x}^*_{\hat{A}}) - F(\mathbf{b} + \mathbf{x}^*_{\hat{A}})\\
            & = G(k_i \mathbf{e}_i + \mathbf{b}) - G(\mathbf{b}),
        \end{aligned}
    \end{equation*}
    where the inequality comes from the DR-submodularity of $F$. 
    
Next, we show that $G$ is non-decreasing on $ [0, \mathbf{x}^*_A]$.  If there exist $\mathbf{y} \in [0, \mathbf{x}^*_A] $ and $\mathbf{y}+\varepsilon \mathbf{e}_i \in [0, \mathbf{x}^*_A]$ such that $G(\mathbf{y} + \varepsilon \mathbf{e}_i) - G(\mathbf{y}) < 0$, i.e., $F(\mathbf{y}+ \varepsilon \mathbf{e}_i + \mathbf{x}^*_{\hat{A}}) - F(\mathbf{y} + \mathbf{x}^*_{\hat{A}}) < 0$, then $F(\varepsilon \mathbf{e}_i| \mathbf{y}+\mathbf{x}^*_{\hat{A}}) \le 0$. According to the DR-submodularity of $F$, we have $F(\varepsilon \mathbf{e}_i|\mathbf{x}^* - \varepsilon \mathbf{e}_i) \le F(\varepsilon \mathbf{e}_i| \mathbf{y}+\mathbf{x}^*_{\hat{A}}) < 0 $, implying that $F(\mathbf{x}^* - \varepsilon \mathbf{e}_i) > F(\mathbf{x}^*)$, which contradicts the optimality of $\mathbf{x}^*$.

From $\mathbf{w}_A \wedge \mathbf{x}^*_A \in [\mathbf{0}, \mathbf{x}^*_A]$ and $w_i = m, \forall i \in A$, we have $\mathbf{w}_A \wedge \mathbf{x}^*_A = m \mathbf{x}^*_A + (1-m) (\mathbf{x}^*_A - \frac{1}{1-m}(\mathbf{x}^*_A - \mathbf{x}^*_A \wedge \mathbf{w}_A))$. Also, $\mathbf{x}^*_A - \frac{1}{1-m}(\mathbf{x}^*_A - \mathbf{x}^*_A \wedge \mathbf{w}_A) = \frac{1}{1-m} \mathbf{x}^*_A \wedge \mathbf{w}_A -\frac{m}{1-m}\mathbf{x}^*_A \ge 0$. For $i \in A$, if $m \le x^*_i$, we have $\frac{m-m x^*_i}{1-m} \le x^*_i \iff m-m x^*_i \le x^*_i -mx^*_i$. For $i \in A$, if $m \ge x^*_i$, we have $\frac{x^*_i-mx^*_i}{1-m} \le x^*_i \iff x^*_i-m x^*_i \le x^*_i -mx^*_i$. It implies that $\mathbf{x}^*_A - \frac{1}{1-m}(\mathbf{x}^*_A - \mathbf{x}^*_A \wedge \mathbf{w}_A) \le \mathbf{x}^*_A$, i.e., $\mathbf{x}^*_A - \frac{1}{1-m}(\mathbf{x}^*_A - \mathbf{x}^*_A \wedge \mathbf{w}_A) \in [\mathbf{0}, \mathbf{x}^*_A]$. Then, we have
    \begin{equation*}
        \begin{aligned}
            G(\mathbf{w}_A \wedge \mathbf{x}^*_A) & \ge m G(\mathbf{x}^*_A) + (1-m) G\left(\mathbf{x}^*_A - \frac{1}{1-m}(\mathbf{x}^*_A - \mathbf{x}^*_A \wedge \mathbf{w}_A)\right)\\
            & \ge m G(\mathbf{x}^*_A) = m F(\mathbf{x}^*) - m F(\mathbf{x}^*_{\hat{A}}),
        \end{aligned}
    \end{equation*}
 where the first inequality comes from Lemma \ref{lemma 3}; and the second inequality comes from Lemma \ref{lemma 4}. 
 
 Finally we complete our proof.
    \begin{equation*}
        \begin{aligned}
            F(\mathbf{x} \vee \mathbf{z}') & = F(\mathbf{x} \vee (\mathbf{x}^* \wedge \mathbf{w}) ) \ge (1-m) F(\mathbf{w}\wedge \mathbf{x}^*) = (1-m) F((\mathbf{w}_A\wedge\mathbf{x}^*_A) \vee \mathbf{x}^*_{\hat{A}}) \\
            & = (1-m) \left[F(\mathbf{x}^*_{\hat{A}}) +  F(\mathbf{w}_A\wedge\mathbf{x}^*_A | \mathbf{x}^*_{\hat{A}}) \right] = (1-m) \left[F(\mathbf{x}^*_{\hat{A}}) + G(\mathbf{w}_A\wedge\mathbf{x}^*_A)\right] \\
            & \ge (1-m) \left[F(\mathbf{x}^*_{\hat{A}}) + m G (\mathbf{x}^*_A)\right]  = (1-m) \left[F(\mathbf{x}^*_{\hat{A}}) + m F(\mathbf{x}^*) - m F(\mathbf{x}^*_{\hat{A}})\right] \\
            & = (1-m)m F(\mathbf{x}^*) + (1-m)^2 F(\mathbf{x}^*_{\hat{A}}),
        \end{aligned}
    \end{equation*}
    where the inequality comes from the DR-submodularity of $F$ and $G$, and Lemma \ref{lemma 4}.
\end{proof}

\begin{lemma}\label{lemma 12}
    Let $\mathbf{x}\in \mathcal{P}\cap[0,m]^n$ be a stationary point in the reduced feasible domain and $\mathbf{x}^*\in \mathcal{P}$ be the optimal solution. Denote $A=\{i \in [n]: x_i \ge m-\varepsilon\}$ and $\hat{A} = [n]\setminus A$. For $\mathbf{z}' = \mathbf{x}^* \wedge \mathbf{w}$, we have:
    \begin{equation}
       F(\mathbf{x} \wedge \mathbf{z}') \ge (m- \varepsilon) \left[F(\mathbf{x}^*) - F(\mathbf{x}^*_{\hat{A}}) \right]
    \end{equation}
\end{lemma}

\begin{proof}
    \begin{equation*}
        \begin{aligned}
            F(\mathbf{x} \wedge \mathbf{z}') & = F(\mathbf{x} \wedge \mathbf{x}^* \wedge \mathbf{w}) = F(\mathbf{x} \wedge \mathbf{x}^*) = F((\mathbf{x}_A \vee \mathbf{1}_{\hat{A}}) \wedge (\mathbf{x}_{\hat{A}} \vee \mathbf{1}_{A}) \wedge \mathbf{x}^* ) \\
            & \ge (m- \varepsilon) F((\mathbf{x}_{\hat{A}} \vee \mathbf{1}_{A}) \wedge \mathbf{x}^* ) = (m- \varepsilon) F((\mathbf{x}_{\hat{A}} \wedge \mathbf{x}^*_{\hat{A}}) \vee \mathbf{x}^*_{A} ) \\
            & \ge (m- \varepsilon) \left[F(\mathbf{x}^*) - F(\mathbf{x}^*_{\hat{A}}) \right].
        \end{aligned}
    \end{equation*}
    In the above, the first inequality is because of Lemma \ref{lemma 4}. The second inequality comes from the submodularity $F((\mathbf{x}_{\hat{A}} \wedge \mathbf{x}^*_{\hat{A}}) \vee \mathbf{x}^*_{A} )  + F(\mathbf{x}^*_{\hat{A}}) \ge F(\mathbf{x}^*) + F((\mathbf{x}_{\hat{A}} \vee \mathbf{x}^*_{A}) \wedge \mathbf{x}^*_{\hat{A}})$, where $(\mathbf{x}_{\hat{A}} \wedge \mathbf{x}^*_{\hat{A}}) \vee \mathbf{x}^*_A \vee \mathbf{x}^*_{\hat{A}} = \mathbf{x}^*$ and $((\mathbf{x}_{\hat{A}} \wedge \mathbf{x}^*_{\hat{A}}) \vee \mathbf{x}^*_A) \wedge \mathbf{x}^*_{\hat{A}}  = (\mathbf{x}_{\hat{A}} \vee \mathbf{x}^*_{A}) \wedge (\mathbf{x}^*_{\hat{A}} \vee \mathbf{x}^*_{A}) \wedge \mathbf{x}^*_{\hat{A}}$.
\end{proof}

In light of the above lemmas, next, we prove the approximation ratio of stationary points in $\mathcal{P}\cap [0,m]^n$. 

\begin{theorem}\label{thm:0.309}
    Any stationary point in $\mathcal{P}\cap \left[0,\frac{3-\sqrt{5}}{2}\right]^n$ produces a value that is at least $(0.309-O(\varepsilon))$ of the optimal solution for problem (\ref{eq:main_prob}).
\end{theorem}
\begin{proof}
    According to Lemmas \ref{lemma 10}, \ref{lemma 11}, and \ref{lemma 12}, we have: 
    \begin{equation*}
        \begin{aligned}
            F(\mathbf{x}) & \ge \cfrac{1}{2} \left[F(\mathbf{x}\vee \mathbf{z}') + F(\mathbf{x} \wedge \mathbf{z}')\right] \\
            & \ge \cfrac{1}{2} \left[(1-m)m F(\mathbf{x}^*) + (1-m)^2 F(\mathbf{x}^*_{\hat{A}}) + (m- \varepsilon) F(\mathbf{x}^*) - (m -\varepsilon)F(\mathbf{x}^*_{\hat{A}}) \right] \\
            & = \cfrac{(2-m)m -\varepsilon}{2} F(\mathbf{x}^*) + \cfrac{(1-m)^2-m +\varepsilon}{2} F(\mathbf{x}^*_{\hat{A}})
        \end{aligned}
    \end{equation*}
    Note that $m \in\left[0,\frac{3-\sqrt{5}}{2}\right]$ implies that $\frac{(1-m)^2-m}{2} F(\mathbf{x}^*_{\hat{A}}) \ge 0$. So choosing 
    \[m^*= \arg\max_{m \in \left[0,\frac{3-\sqrt{5}}{2}\right]}\frac{(2-m)m}{2} = \frac{3-\sqrt{5}}{2}\]
     leads to $\frac{(2-m^*)m^*}{2} \ge 0.309$.
\end{proof}

Efficient algorithms exist to return a stationary point in $\mathcal{P} \cap [0,m]^n$. For example, the Frank-Wolfe algorithm with the adaptive step (\cite{Lacoste-Julien}) returns a stationary point with a gap smaller than $\varepsilon$ in at most $O\left(\frac{1}{\varepsilon^2}\right)$ iterations. Hence Theorem~\ref{thm:0.309} with this Frank-Wolfe algorithm offers us a $(0.309-O(\varepsilon))$ approximation algorithm for problem (\ref{eq:main_prob}). 

To understand why this restricted local search offers better approximation ratio than the vanilla version of the local search, note that points stationary in the interior of $\mathcal{P} \cap [0,m]^n$ are also  stationary in $\mathcal{P}$; but not the other way around. Therefore, the reduced feasible domain removes some boundary stationary points in $\mathcal{P}$, and changes some non-stationary points in $\mathcal{P}$ to stationary points in the reduced feasible domain, leading to better performances. 

Recall our previous Example \ref{example 1}.  In Tab.\ref{Ratio1}, we report the value ratios for the worst stationary points in $\mathcal{P}$ and $\mathcal{P}\cap \left[0,\frac{3-\sqrt{5}}{2}\right]^n$, respectively. The results show that the approximation ratio of stationary points in $\mathcal{P} \cap \left[0,\frac{3-\sqrt{5}}{2}\right]^n$ performs better than those in $\mathcal{P}$. Also, for this particular example, the numerical ratios are  much better than the proven approximation ratio of $0.309$ when $m=\frac{3-\sqrt{5}}{2}$, and leaving open the question on what the exact approximation ratio of the restricted local search is.


\begin{table}
    \centering
    \caption{Ratios between the worst stationary point and the optimal solution}
    \label{Ratio1}
    \begin{tabular}{m{4.5cm} m{1.2cm} m{1.2cm} m{1.2cm} m{1.2cm} m{1.2cm} m{1.2cm} m{1.2cm}}
        \hline
        & $k=2$ & $k=3$ &$k=5$ & $k=10$ & $k=20$ & $k=30$ & $k=50$ \\
        \hline
        Restricted Continuous Local Search with $m=1$ & $0.4142$ & $0.3222$ & $0.1999$ & $0.1000$ & $0.0500$ & $0.0333$ & $ 0.0200$\\
        \hline
        Restricted Continuous Local Search with $m=\frac{3-\sqrt{5}}{2}$& $0.4271$ & $0.4062$ & $0.3945$ & $0.3877$ & $0.3847$ & $0.3838$ & $0.3830$\\
        \hline
    \end{tabular}
\end{table}

\section{Aided Frank-Wolfe Variant}\label{section 4}
Because the proof of the approximation \emph{ratio of the aided measured continuous greedy} algorithm in \cite{Buchbinder2} strongly depends on the discrete optimal solution and the inequality that the multilinear extension of any submodular set function is bounded from below by its Lovasz extension, whether the same algorithm yields a $0.385$ approximation ratio for problem (\ref{eq:main_prob}) on the continuous domain admits is still an open problem. In this section, we answer this question affirmatively. Because our proof only utilizes DR-submodularity instead of the multilinear extension, the Lovasz extension, and the discrete optimal solution, it brings out more flexibility in tuning the parameters used in the algorithm. To design and discretize our algorithm, we adopt the systematic framework proposed in \cite{Du-submodular-book2022}, namely Lyapunov framework. 

Firstly, we prove the following essential lemma which will be used to derive an upper bound on the optimal solution to clear the way to adopt the aforementioned Lyapunov framework. 
Note that the lemma below extends Lemma 4.4 in \cite{Buchbinder2} from binary optimal solution to any solution in $\mathcal{P}$ and its proof only uses DR-submodularity and non-negativity of $F$. Denote
\begin{eqnarray*}
\theta_A(t) &:=&\left\| \mathbf{x}_A(t) \right\|_{\infty}=\max_{i\in A} \mathbf{x}_i(t);\quad \theta_{\hat{A}}(t):=\left\| \mathbf{x}_{\hat{A}}(t) \right\|_{\infty}=\max_{i\in \hat{A}} \mathbf{x}_i(t)
\end{eqnarray*}
\begin{lemma}\label{lemma 14}
    For any $\mathbf{z} \in \mathcal{P}$ and a partition $A \cap \hat{A} = \emptyset$, $A \cup \hat{A} = [n]$, assume $\theta_A(t)\le \theta_{\hat{A}}(t) $, we have:
    \begin{equation}
        F(\mathbf{x}(t) \vee \mathbf{z}) \ge \left(1-\theta_A(t)\right) F(\mathbf{z}) - \left(\theta_{\hat{A}}(t) -\theta_A(t)\right) F(\mathbf{z}\vee \mathbf{1}_A)
    \end{equation} 
\end{lemma}
\begin{proof}
    For convenience and without ambiguity, denote $\mathbf{x} = \mathbf{x}(t)$, $\theta_A =\theta_A(t)$, and $\theta_{\hat{A}} = \theta_{\hat{A}}(t)$. First, we have: 
    \[
    \mathbf{x} \vee \mathbf{z} = \left(1-\theta_{\hat{A}}\right) \mathbf{z} + \theta_{\hat{A}} \left[\mathbf{z}+\frac{1}{\theta_{\hat{A}}}(\mathbf{x}\vee \mathbf{z}-\mathbf{z})\right].
    \] Then,
    \begin{equation*}
    \mathbf{q} = \mathbf{z}+\cfrac{1}{\theta_{\hat{A}}}(\mathbf{x}\vee \mathbf{z}-\mathbf{z}) = \cfrac{\mathbf{x}\vee \mathbf{z}-\left(1-\theta_{\hat{A}}\right)\mathbf{z}}{\theta_{\hat{A}}}  \begin{cases} = z_i &  z_i \ge x_i \\ \le \cfrac{\theta_A}{\theta_{\hat{A}}} &  x_i > z_i, i \in A \\ \le \cfrac{\theta_A}{\theta_{\hat{A}}} &  x_i > z_i, \text{ some } i \in \hat{A} \\ > \cfrac{\theta_A}{\theta_{\hat{A}}}  &  x_i > z_i, \text{ others }i \in \hat{A} \end{cases}
    \end{equation*}
    According to the above inequalities, define three sets 
    \begin{eqnarray*}
    B_1&=&\left\{i \in [n] : z_i \ge x_i\right\}\\
    B_2&=&\left\{i\in [n] : x_i > z_i, q_i \le \frac{\theta_A}{\theta_{\hat{A}}}\right\}\\
    B_3 &=&\left\{i \in [n] : x_i > z_i, 1 \ge q_i > \frac{\theta_A}{\theta_{\hat{A}}}\right\}
    \end{eqnarray*}
    Obviously, $B_3 \subseteq \hat{A}$ and $\mathbf{q}_{B_2} \ge \mathbf{z}_{B_2}$, implying that $\mathbf{q} = \mathbf{q}_{B_1} \vee \mathbf{q}_{B_2} \vee \mathbf{q}_{B_3} \vee \mathbf{z}_{B_2}$. Then, we have:
    \begin{equation*}
        \begin{aligned}
            F(\mathbf{x}(t) \vee \mathbf{z}) & \ge \left(1-\theta_{\hat{A}}\right) F(\mathbf{z}) + \theta_{\hat{A}} F\left(\mathbf{z}+\frac{1}{\theta_{\hat{A}}}(\mathbf{x}\vee \mathbf{z}-\mathbf{z})\right) \\
            & = \left(1-\theta_{\hat{A}}\right) F(\mathbf{z}) + \theta_{\hat{A}} F(\mathbf{q}_{B_1} \vee \mathbf{q}_{B_2} \vee \mathbf{q}_{B_3} \vee \mathbf{z}_{B_2}) \\
            & \ge \left(1-\theta_{\hat{A}}\right) F(\mathbf{z}) + \theta_{\hat{A}}\cfrac{\theta_{\hat{A}}-\theta_A}{\theta_{\hat{A}}} F(\mathbf{q}_{B_1} \vee \mathbf{q}_{B_3} \vee \mathbf{z}_{B_2}) \\
            & \ge \left(1-\theta_{\hat{A}}\right) F(\mathbf{z}) + \left(\theta_{\hat{A}}-\theta_A\right) (F(\mathbf{z}) - F(\mathbf{z} \vee \mathbf{1}_A) ) \\
            & = \left(1-\theta_A\right) F(\mathbf{z}) - \left(\theta_{\hat{A}}-\theta_A\right) F(\mathbf{z}\vee \mathbf{1}_A),
        \end{aligned}
    \end{equation*}
    where the first inequality comes from Lemma \ref{lemma 3}; the second inequality comes from Lemma \ref{lemma 4}; and the third inequality comes from the DR-submodularity and non-negativity of $F$, i.e., $F(\mathbf{q}_{B_1} \vee \mathbf{z}_{B_2} \vee \mathbf{q}_{B_3}) + F(\mathbf{z} \vee \mathbf{1}_A) \ge F(\mathbf{z}) + F(\mathbf{z}\vee \mathbf{1}_A \vee \mathbf{q}_{B_1} \vee \mathbf{z}_{B_2} \vee \mathbf{q}_{B_3})$ .
\end{proof}

\subsection{Lyapunov Function}\label{subsection 4.1}
In this subsection, we utilize the following parametric Lyapunov function $E(\mathbf{x}(t))$ to design our algorithm for problem (\ref{eq:main_prob}). For convenience, we denote $\dot{E}(\mathbf{x}(t)) = \frac{d E(\mathbf{x}(t))}{d t}$.
\begin{definition} \emph{(\cite{Du-submodular-book2022})}
    The Lyapunov function associated with a continuous-time algorithm $\mathbf{x}(t)$ and a feasible solution $\mathbf{y}\in\mathcal{P}$ for problem (\ref{eq:main_prob}) is as follows: 
    \begin{equation}\label{lyapunov}
        E(\mathbf{x}(t)) = a(t) F(\mathbf{x}(t)) + h(t) F(\mathbf{y}) - d(t) F(\mathbf{x}^*)
    \end{equation}
    And the upper bound of the optimal solution $F(\mathbf{x}^*)$ is as follows:
    \begin{equation}\label{eq:up-bound}
        F(\mathbf{x}^*) \le p(\mathbf{x}(t)) F(\mathbf{x}(t)) + q(\mathbf{x}(t)) \langle \nabla F(\mathbf{x}(t)) , v(\mathbf{x}(t)) \rangle + r(\mathbf{x}(t)) F(\mathbf{y})
    \end{equation}
    where $a(t)$, $h(t)$, $d(t)$ are non-negative and non-decreasing, $p(\mathbf{x}(t))$, $q(\mathbf{x}(t))$, $r(\mathbf{x}(t))$ are non-negative. 
\end{definition}

We recall the main ingredients of the Lyapunov framework introduced by \cite{Du1}. The key idea is to enforce monotonicity of the Lyapunov function $E(\mathbf{x}(t))$, implying that $E(\mathbf{x}(t)) \ge E(\mathbf{x}(0))$, from which the approximation ratio of the algorithm $\mathbf{x}(t)$ follows:
\begin{equation}\label{approximation ratio}
   \cfrac{a(t) F(\mathbf{x}(t)) }{a(t) + h(t) - h(0)} +  \cfrac{\left[h(t)-h(0)\right]F(\mathbf{y})}{a(t) + h(t) - h(0)}  \ge \cfrac{d(t)-d(0)}{a(t)+h(t) - h(0)} F(\mathbf{x}^*) +\cfrac{a(0)}{a(t)+h(t) - h(0)} F(\mathbf{x}(0))
\end{equation}
It implies that the continuous algorithm returns the solution $\mathbf{x}(t)$ with probability $p = \frac{a(t)}{a(t)+h(t)-h(0)}$ and returns the solution $\mathbf{y}$ with probability $1-p$.

To achieve non-decreasing $E(\mathbf{x}(t))$ without using any information on the optimal value, we backtrack sufficient conditions for the following to hold:
\begin{eqnarray*}
         \dot{E}(\mathbf{x}(t)) & =& \dot{a}(t) F(\mathbf{x}(t)) + a(t) \langle \nabla F(\mathbf{x}(t)), \dot{\mathbf{x}}(t)) \rangle +\dot{h}(t) F(\mathbf{y}) - \dot{d}(t) F(\mathbf{x}^*) \ge 0
 \end{eqnarray*}
Replacing $F(\mathbf{x}^*)$ with an upper bound in (\ref{eq:up-bound}) to obtain a lower bound of $\dot{E}(\mathbf{x}(t))$. A sufficient condition for the above is that this resultant lower bound is non-negative. 
\begin{eqnarray*}
         \dot{E}(\mathbf{x}(t)) &\overset{(\ref{eq:up-bound})}{\ge} &(\dot{a}(t)- \dot{d}(t) p(\mathbf{x}(t)) ) F(\mathbf{x}(t)) + (\dot{h}(t)- \dot{d}(t) r(\mathbf{x}(t))) F(\mathbf{y})\\
         && \langle \nabla F(\mathbf{x}(t)), a(t) \dot{\mathbf{x}}(t)- \dot{d}(t) q(\mathbf{x}(t)) v(\mathbf{x}(t))\rangle  \ge 0
   \end{eqnarray*}
Because $F(\mathbf{x}(t))$ and $F(\mathbf{y})$ are non-negative, and $F(\mathbf{x}(t))$ is non-monotone, a sufficient condition for the above is that the coefficients of the first two terms to be non-negative and the coefficient of the third term is equal to zero:
\begin{eqnarray*}
        \dot{a}(t) & \ge&  \dot{d}(t) p(\mathbf{x}(t)), \dot{h}(t) \ge \dot{d}(t) r(\mathbf{x}(t)),\\
         \dot{\mathbf{x}}(t) & =& \cfrac{\dot{d}(t) q(\mathbf{x}(t)) v(\mathbf{x}(t))}{a(t)}
   \end{eqnarray*}
   
%
%
%

Hence, we can construct the following variation problem to design our continuous algorithm $\mathbf{x}(t)$ such that the approximation ratio is maximized.
\begin{equation} \label{variation problem in lyapunov}
    \begin{aligned}
        \max \quad & \cfrac{d(t)-d(0)}{a(t) + h(t) - h(0)} \\
        \text{s.t.} \quad & \dot{a}(t)  \ge  \dot{d}(t) p(\mathbf{x}(t))\\
        & \dot{h}(t) \ge \dot{d}(t) r(\mathbf{x}(t)) \\
        & \dot{\mathbf{x}}(t))  = \cfrac{\dot{d}(t) q(\mathbf{x}(t)) v(\mathbf{x}(t))}{a(t)} \\
        & a(t),h(t),d(t) \text{ are non-negative and non-decreasing}
    \end{aligned}
\end{equation}
One of the biggest advantages of this Lyapunov framework is that the algorithm design and analysis process can be put on autopilot mode once we specify the parametric forms of the Lyapunov function and the upper bound on the optimal value.




\subsection{Extending the algorithm in \cite{Buchbinder2} via Lyapunov framework}
Lemma \ref{lemma 14} holds for all feasible solutions. In other words, given any map $\mathcal{P}\ni \mathbf{x}^*\mapsto g(\mathbf{x}^*) \in \mathcal{P}$, we obtain an upper bound variant on the optimal solution as follows:
\begin{equation} 
    F(g(\mathbf{x}^*)) \le \cfrac{F(\mathbf{x}(t)) + \langle \nabla F(\mathbf{x}(t)) , \mathbf{x}(t) \vee g(\mathbf{x^*}) - \mathbf{x}(t) \rangle  + \left(\theta_{\hat{A}}(t) -\theta_A(t)\right) F(g(\mathbf{x}^*)\vee \mathbf{1}_{A})}{1-\theta_A(t)} 
\end{equation}
Therefore, the main challenge to obtain a good upper bound is to construct a proper map $g(\mathbf{x}^*)$. The down-closeness of the convex $\mathcal{P}$ suggests a natural map: $g(\mathbf{x}^*) = \pi(t) \wedge \mathbf{x}^*$,  for any given vector $\pi(t) \in [0,1]^n$. Note that $g(\mathbf{x}^*) \le \mathbf{x}^* \in \mathcal{P}$. Different choices of $\pi(t)$ will produce different algorithms. The one in \cite{Buchbinder2} is as follows:
\begin{equation}\label{eq:pi}
    \pi(t) = \pi_A(t) \mathbf{1}_A + \pi_{\hat{A}}(t) \mathbf{1}_{\hat{A}}
\end{equation}
where $\pi_A(t) = 0$ if $t \le \theta$ and $\pi_A(t) = 1$ if $t > \theta$; and  $\pi_{\hat{A}}(t)=1, \forall t\in [0,1]$, where $A\subseteq [n]$ is a random set that depends on the stationary point $\mathbf{y}$: namely, each $i \in [n]$ is included independently in $A$ with probability $y_i$. If we fix both the parametric form of $\pi(t)$ as in (\ref{eq:pi}) and the choice of $A$, but allow $\pi_A(t)$ and $\pi_{\hat{A}}(t)$ to choose other continuous functions, can we get better approximation ratio? It turns out the answer is negative as will be explained now via the Lyapunov framework. From this discussion, we propose the following Algorithm~\ref{algorithm 1}.

\begin{algorithm}
\caption{}
	{\bf Input:} $F$: objective function, $\mathcal{P}$: feasible domain, $\pi(t) \in [0,1]$: parameter function.\\
    {\bf Output:} the solution $\mathbf{x}(1)$ with probability $1-p$ and the solution $\mathbf{y}$ with probability $p$
 
	\begin{algorithmic}[1]
        
		\STATE $\mathbf{y} \leftarrow$ Frank-Wolfe Algorithm $(F,\mathcal{P})$
        \STATE $A$ is a random subset of $[n]$. For all $i\in [n]$, $i \in A$  with probability of $y_i$ independently.\\
        \text{/* \quad Frank-Wolfe Variant Stage \quad */}
        \STATE Initialize $\mathbf{x}(0) = \mathbf{0}$ 
		\FOR{ $t \in [0,1]$}
        \STATE  $\mathbf{v}(t) = \arg\max\limits_{\mathbf{v} \in \mathcal{P}, \mathbf{v} \le (1-\mathbf{x}(t))\odot \pi(t)}  \langle\nabla F(\mathbf{x}(t)), \mathbf{v} \rangle$
        \STATE $\frac{d \mathbf{x}(t)}{d t} =  \mathbf{v}(t)$
        \ENDFOR
	\end{algorithmic}
    \label{algorithm 1}
\end{algorithm}

The critical thing to employ the Lyapunov framework is to find a good upper bound with a parametric form as in (\ref{eq:up-bound}). Next, we will take on this job.

\begin{lemma}\label{lemma 16}
    For each $t \in [0,1]$, we have: 
    \begin{equation}\label{variation problem bound}
        \begin{aligned}
            F(\mathbf{x}(t) \vee (\pi(t) \wedge \mathbf{x}^*))&  \ge \pi_{\hat{A}}(t)\cdot (1-\theta_A(t)) F(\mathbf{x}^*) - (1-\pi_A(t)) \cdot (1-\theta_A(t)) F(\mathbf{x}^* \wedge \mathbf{1}_A)\\
            & \quad -  (\theta_{\hat{A}}(t) - \theta_A(t)) F( (\pi(t)\wedge \mathbf{x}^*) \vee \mathbf{1}_A)
        \end{aligned}
    \end{equation}
\end{lemma}
\begin{proof}
    \begin{equation*}
        \begin{aligned}
            F(\mathbf{x}(t) \vee (\pi(t) \wedge \mathbf{x}^*)) & \ge (1-\theta_A(t)) F(\pi(t) \wedge \mathbf{x}^*) - (\theta_{\hat{A}}(t)-\theta_A(t)) F((\pi(t) \wedge \mathbf{x}^*) \vee \mathbf{1}_A)  \\
            & \ge (1-\theta_A(t)) [F(\mathbf{x}^*\wedge \pi(t) \wedge \mathbf{1}_A) + F(\mathbf{x}^* \wedge (\mathbf{1}_A \vee \pi_{\hat{A}}(t) \mathbf{1}_{\hat{A}})) - F(\mathbf{x}^* \wedge \mathbf{1}_A)] \\
            & \quad - (\theta_{\hat{A}}(t) - \theta_A(t))F( (\pi(t) \wedge \mathbf{x}^*) \vee \mathbf{1}_A) \\
            & \ge \pi_{\hat{A}}(t)\cdot (1-\theta_A(t)) F(\mathbf{x}^*) - (1-\pi_A(t)) \cdot (1-\theta_A(t)) F(\mathbf{x}^* \wedge \mathbf{1}_A) \\
            & \quad - (\theta_{\hat{A}}(t) - \theta_A(t)) F( (\pi(t) \wedge \mathbf{x}^*) \vee \mathbf{1}_A)
        \end{aligned}
    \end{equation*}
    The first inequality is because of Lemma \ref{lemma 14}. The second inequality comes from submodularity, i.e., $F(\pi(t) \wedge \mathbf{x}^*) + F(\mathbf{x}^* \wedge \mathbf{1}_A) \ge F(\mathbf{x}^* \wedge \mathbf{1}_A \wedge \pi(t)) + F(\mathbf{x}^* \wedge (\mathbf{1}_A\vee \pi(t)))$. The last inequality is due to Lemma \ref{lemma 4}. 
\end{proof}

Lemma \ref{lemma 16} implies the following upper bound on the optimal value $F(\mathbf{x}^*)$.
\begin{equation}\label{ upper bound of x^*}
\begin{array}{c}
 \begin{aligned}
        F(\mathbf{x}^*) & \le \cfrac{F(\mathbf{x}(t))}{(1-\theta_A(t))\pi_{\hat{A}}(t)} + \cfrac{\langle \nabla F(\mathbf{x}(t)) , \mathbf{x}(t) \vee (\pi(t) \wedge \mathbf{x}^*)   - \mathbf{x}(t) \rangle}{(1-\theta_A(t))\pi_{\hat{A}}(t)} \\
        & \quad + \cfrac{\left(\theta_{\hat{A}}(t) -\theta_A(t)\right) }{(1-\theta_A(t))\pi_{\hat{A}}(t)} F((\mathbf{x}^*\wedge \pi(t))\vee \mathbf{y}) + \cfrac{1-\pi_A(t)}{\pi_{\hat{A}}(t)} F(\mathbf{x}^*\wedge \mathbf{y}) \\
        & \le  \cfrac{\langle \nabla F(\mathbf{x}(t)) , \mathbf{x}(t) \vee (\pi(t) \wedge \mathbf{x}^*)   - \mathbf{x}(t) \rangle}{(1-\theta_A(t)) \pi_{\hat{A}}(t)} +\cfrac{F(\mathbf{x}(t))}{(1-\theta_A(t))\pi_{\hat{A}}(t)}+ r(\mathbf{x}(t))F(\mathbf{y}) \\ 
        & = \cfrac{\langle \nabla F(\mathbf{x}(t)) ,\pi(t) \wedge \mathbf{x}^*  - \mathbf{x}(t) \wedge \pi(t) \wedge \mathbf{x}^*\rangle}{(1-\theta_A(t)) \pi_{\hat{A}}(t)} +\cfrac{F(\mathbf{x}(t))}{(1-\theta_A(t))\pi_{\hat{A}}(t)}+ r(\mathbf{x}(t))F(\mathbf{y}) \\
        & \le \max_{\mathbf{y} \in \mathcal{P}, \mathbf{v} \le \pi(t) \wedge \mathbf{y}  - \mathbf{x}(t) \wedge \pi(t) \wedge \mathbf{y}} \cfrac{\langle \nabla F(\mathbf{x}(t)), \mathbf{v} \rangle}{(1-\theta_A(t)) \pi_{\hat{A}}(t)} +\cfrac{F(\mathbf{x}(t))}{(1-\theta_A(t)) \pi_{\hat{A}}(t)}+ r(\mathbf{x}(t))F(\mathbf{y})  \\
        & \le \max_{\mathbf{v} \in \mathcal{P}, \mathbf{v} \le (1-\mathbf{x}(t))\odot  \pi(t)} \cfrac{\langle \nabla F(\mathbf{x}(t)), \mathbf{v} \rangle}{(1-\theta_A(t)) \pi_{\hat{A}}(t)} +\cfrac{F(\mathbf{x}(t))}{(1-\theta_A(t)) \pi_{\hat{A}}(t)}+ r(\mathbf{x}(t))F(\mathbf{y}) 
    \end{aligned} \\
   \\ 
   \bigg \Downarrow \\
    p(\mathbf{x}(t))= q(\mathbf{x}(t))= \cfrac{1}{(1-\theta_A(t)) \pi_{\hat{A}}(t)}\\
    v(\mathbf{x}(t)) = \arg\max\limits_{\mathbf{v} \in \mathcal{P}, \mathbf{v} \le (1-\mathbf{x}(t))\odot  \pi(t)} \langle \nabla F(\mathbf{x}(t)), \mathbf{v} \rangle \\
    r(\mathbf{x}(t)) = \cfrac{\left(\theta_{\hat{A}}(t) -\theta_A(t)\right)}{(1-\theta_A(t))\pi_{\hat{A}}(t)} + \cfrac{1- \pi_A(t)}{\pi_{\hat{A}}(t) \left\| \pi(t) \right \|_{-\infty}  }
\end{array}
\end{equation}
The first inequality comes from Eq.(\ref{variation problem bound}), Lemmas \ref{lemma 3}, and \ref{lemma 5}, i.e., $F((\mathbf{x}^* \wedge \pi(t) )\vee \mathbf{x}(t)) \le F(\mathbf{x}(t)) + \langle \nabla F(\mathbf{x}(t)), (\mathbf{x}^* \wedge \pi(t) )\vee \mathbf{x}(t) - \mathbf{x}(t) \rangle$. The second inequality is from Lemma \ref{lemma 8}, i.e., $F(\mathbf{y}) = \int_0^1 F(\mathbf{y}) d \lambda(t) \ge \gamma \int_0^1 F((\pi(t) \wedge \mathbf{x}^*) \vee \mathbf{y}) d \lambda(t) + (1-\gamma)\int_0^1 F(\pi(t) \wedge \mathbf{x}^* \wedge \mathbf{y}) d \lambda(t) \ge \gamma \int_0^1 F((\pi(t) \wedge \mathbf{x}^*) \vee \mathbf{y}) d \lambda(t) + (1-\gamma)F(\mathbf{x}^* \wedge \mathbf{y})\int_0^1 \left\|\pi(t)\right\|_{-\infty} d \lambda(t) $ for every $\int_0^1 d\lambda(t) = 1$. The last inequality is because of $\pi(t) \wedge \mathbf{y}  - \mathbf{x}(t) \wedge \pi(t) \wedge \mathbf{y} \le (1-\mathbf{x}(t)) \odot (\pi(t) \wedge \mathbf{y}) \le (1-\mathbf{x}(t)) \odot \pi(t)$. Therefore, for the given constraints, we can decide $a(t)$, $d(t)$, $h(t)$, $\theta_A(t)$, $\theta_{\hat{A}}(t)$ by the following equations. 

\begin{equation}\label{eq:para}
     \begin{cases}
        &  \dot{a}(t)  =  \dot{d}(t) p(\mathbf{x}(t))\\
        & \dot{h}(t) = \dot{d}(t) r(\mathbf{x}(t)) \\
        & \dot{\mathbf{x}}(t))  = \cfrac{\dot{d}(t) q(\mathbf{x}(t)) v(\mathbf{x}(t))}{a(t)} \\
    \end{cases}  
    \Rightarrow
    \begin{cases}
        &  a(t)  = e^t\\
        & \dot{\mathbf{x}}(t)  =  v(\mathbf{x}(t)) \\
        & \theta_A(t) = 1- e^{-\int_0^t \pi_A(s)ds} \\
        & \theta_{\hat{A}}(t) = 1- e^{-\int_0^t \pi_{\hat{A}}(s)ds} \\
        & d(t) - d(0) = \int_0^t g_A(s) \pi_{\hat{A}}(s)  ds \\
        & h(t) - h(0) = \int_0^t g_A(s) - g_{\hat{A}}(s) + \frac{1-\pi_A(s)}{\left\| \pi(s)\right\|_{-\infty}} g_A(s) ds
    \end{cases}
\end{equation}
Because of $p(\mathbf{x}(t)) = q(\mathbf{x}(t))$, we have $\dot{\mathbf{x}}(t) = \frac{\dot{a}(t)}{a(t)} v(\mathbf{x}(t))$. To ensure that the final output $\mathbf{x}(1) = \int_0^1 v(\mathbf{x}(s)) d (\ln a(s))$ is a feasible solution in $\mathcal{P}$, we fix the stopping time at $t=1$ and hence $\mathbf{x}(t)$ is a convex combination of feasible solutions $v(\mathbf{x}(s))$. It induces that we have $\ln (a(0)) = 0$ and $\ln(a(1)) = 1$, i.e., $a(t)=e^t$. Also, we have $\dot{\mathbf{x}}(t) \le \frac{\dot{a}(t)}{a(t)} \odot (\mathbf{1}-\mathbf{x}(t)) \odot \pi(t)$. It implies $x_i(t) \le 1- {e^{-\int_0^t \pi_A(s)ds}}$ for $i \in A$ and $x_i(t) \le 1- {e^{-\int_0^t \pi_{\hat{A}}(s)ds}}$ for $i \in \hat{A}$. Therefore, we have the explicit form of $\theta_A(t),\theta_{\hat{A}}(t), d(t), h(t)$. For convenience, we denote $g_A(t) = e^{t - \int_0^t \pi_A(s) ds}$ and $g_{\hat{A}}(t) = e^{t - \int_0^t \pi_{\hat{A}}(s) ds}$. 

Plugging Eq.(\ref{ upper bound of x^*}) and Eq.(\ref{eq:para}) into Eq.(\ref{variation problem in lyapunov}), we construct the following optimization problem to maximize the approximation ratio of Algorithm~\ref{algorithm 1}. The first constraint comes from the condition of Lemma \ref{lemma 14}, i.e., $\theta_A(t) \le \theta_{\hat{A}}(t)$. The second constraint comes from Lemma \ref{lemma 8}, i.e., $\gamma \in [0,1/2]$.   

\begin{equation}\label{variation problem}
    \begin{aligned}
     \max  \quad &\cfrac{\int_0^1 \pi_{\hat{A}}(q)g_A(q) dq}{e + \int_0^1 g_A(q) - g_{\hat{A}}(q) + \frac{(1-\pi_A(q))g_A(q)}{\left\|\pi(q)\right\|_{-\infty}} dq}   \\
     \text{ s.t. } & e^{- \int_0^t \pi_A(s) ds} \ge e^{- \int_0^t \pi_{\hat{A}}(s) ds} \\
    &  \int_0^1 \cfrac{\theta_{\hat{A}}(t) -\theta_A(t)}{(1-\theta_A(t))\pi_{\hat{A}}(t)} dt \le \int_0^1 \cfrac{1- \pi_A(t)}{\pi_{\hat{A}}(t) \left\| \pi(t) \right \|_{-\infty}  } dt \\
    & \pi_A(t), \pi_{\hat{A}}(t) \in [0,1]
    \end{aligned}
\end{equation}
Note that $1\le g_{\hat{A}}(t) \le g_A(t) \le e^{t}$, $g_{\hat{A}}(t)$ and $g_A(t)$ are non-decreasing. It means that the increasing rate of $\int_0^1 g_A(q) -g_{\hat{A}}(q) dq$ with increasing $\pi_{\hat{A}}(q)$ is less than the that of $\int_0^1 \pi_{\hat{A}}(q) g_A(q) dq $ with increasing $\pi_{\hat{A}}(q)$. In addition, a larger $\pi_{\hat{A}}(q)$ implies a larger $\int_0^1 \pi_{\hat{A}}(q)g_A(q) dq$ and a smaller $\int_0^1\frac{(1-\pi_A(q))g_A(q)}{ \left\|\pi(q)\right\|_{-\infty} }d q$. Hence, a larger $\pi_{\hat{A}}(t)$ induces a larger approximation ratio. Therefore, for all $t \in [0,1]$, we have: 
\begin{equation}\label{pi2}
    \pi_{\hat{A}}(t) = 1 \Rightarrow g_{\hat{A}}(t) = 1
\end{equation}

From $\pi_{\hat{A}}(t) = 1$, it means $F((\pi(t) \wedge \mathbf{x}^*)\vee \mathbf{1}_A) = F(\mathbf{x}^*\vee \mathbf{1}_A)$. In other words, we can replace $F((\pi(t) \wedge \mathbf{x}^*) \vee \mathbf{y})$ by $F(\mathbf{x}^* \vee \mathbf{y})$. Therefore, Eq.(\ref{variation problem bound}) can be rewritten as:

\begin{equation}
    \begin{aligned}
        F(\mathbf{x}(t) \vee (\pi(t) \wedge \mathbf{x}^*))&  \ge (1-\theta_A(t)) F(\mathbf{x}^*) - (1-\pi_A(t)) \cdot (1-\theta_A(t)) F(\mathbf{x}^* \wedge \mathbf{1}_A)\\
        & \quad -  (\theta_{\hat{A}}(t) - \theta_A(t)) F( \mathbf{x}^* \vee \mathbf{1}_A)
    \end{aligned}
\end{equation}

Similar as Eq.(\ref{ upper bound of x^*}) and Eq.(\ref{eq:para}), we can simplify $r(\mathbf{x}(t)) = \frac{\theta_{\hat{A}}(t) -\theta_A(t)}{1-\theta_A(t)} + 1- \pi_A(t)$ and $h(t) - h(0) = \int_0^t g_A(t) - g_{\hat{A}}(t) dt + g_A(t) -g_A(0) $. Therefore, we have a simplified variation problem as follows:

\begin{equation}
    \begin{aligned}
         \max & \quad \cfrac{\int_0^1 g_A(q) dq}{e+ g_A(1)-1 + \int_0^1 g_A(q) - 1 dq } = 1 - \cfrac{1}{1+\frac{\int_0^1 g_A(q) dq}{e+g_A(1)-2}} \\
         \text{ s.t. } & e^{-\int_0^t \pi_A(s)d s} \ge e^{-t}, g_A(t) = e^{t-\int_0^t \pi_A(s) ds} \\ 
         & \int _0^1 \cfrac{1-e^{-t} -\theta_A(t)}{1-\theta_A(t)} dt \le  \int_0^1 1- \pi_A(t) dt   \\
        & \pi_A(t)\in [0,1]
    \end{aligned} 
\end{equation}

It implies that we should make $\int_0^1 g_A(q) dq$ as large as possible and $g_A(1)$ as small as possible. In other words, we hope that $g_A(t)$ is as large as possible for all $t \in [0,1)$ but $g_A(1)$ is as small as possible. Therefore, it represents that there is a $\theta \in [0,1]$ such that $g_A(t)$ has the following form, which uniquely decides the form of $\pi_A(t)$ as follows:
\begin{equation}\label{pi1}
    g_A(t) = e^{t-\int_0^t \pi_A(s) ds}= \begin{cases}
        e^t & t \in [0,\theta]\\
        e^{\theta} & t \in (\theta,1]
    \end{cases} \Rightarrow 
    \pi_A(t)= \begin{cases}
        0 & t \in [0,\theta]\\
        1 & t \in (\theta,1]
    \end{cases}
\end{equation}

\begin{table}[htb]
    \centering
    \caption{Parameter functions in the Lyapunov function}
    \label{Parameter Function in Lyapunov Function}
    \begin{tabular}{m{1.5cm} m{1cm} m{3cm} m{2cm} m{1.5cm} m{1.5cm} m{1.5cm}}
        \hline
        & $a(t)$ & $h(t)$  & $d(t)$ &$p(\mathbf{x}(t))$ & $q(\mathbf{x}(t))$ & $r(\mathbf{x}(t))$ \\
        \hline
        $t \in [0,\theta]$ & $e^t$ & $2 e^t - t$ & $e^{t}$ & $1$ & $1$ &$2-e^{-t}$ \\
        \hline
         $t \in (\theta,1]$ & $e^t$ & $e^{\theta}\cdot (t + 2-\theta) - t$ &$e^\theta\cdot (t+1-\theta)$ &$e^{t-\theta}$ &$e^{t-\theta}$ & $1-e^{-\theta}$\\
        \hline
    \end{tabular}
\end{table}
It is easy to verify $\int _0^1 \frac{1-e^{-t} -\theta_A(t)}{1-\theta_A(t)} dt \le  \int_0^1 1- \pi_A(t) dt$. Then, we obtain the following parametric coefficient functions used in the Lyapunov framework (Eq.(\ref{lyapunov})) in Tab.\ref{Parameter Function in Lyapunov Function}. Therefore, we can reduce the variation problem to a univariate optimization problem:
\begin{equation}
    \max_{\theta \in [0,1]} \cfrac{(2-\theta) \cdot e^{\theta} -1}{ e+3e^\theta -3 -\theta e^{\theta}}
\end{equation}

Solving the above maximization problem, we have $\theta\approx 0.372$, $p\approx 0.77$, and the approximation ratio is at least $0.385$. Therefore, the Lyapunov framework recovers the \emph{aided measured continuous greedy algorithm} in \cite{Buchbinder2} and extended it from the multilinear extension of submodular set functions to general DR-submodular functions. Note that the optimal $\pi_A(t)$ and $\pi_{\hat{A}}(t)$ from the variation problem agree with the ones used in \cite{Buchbinder2}, showing that their choice is optimal under the form of the lower bound in Lemma \ref{lemma 14} and the assumption that components in $A$ are the same, and components in $\hat{A}$ are the same.

\begin{algorithm}
\caption{Discretized Algorithm~\ref{algorithm 1}}
	{\bf Input:} $F$: objective function, $\mathcal{P}$: feasible domain, $\theta \in [0,1]$: parameter, $\varepsilon$: convergence rate, $D$: diameter of the feasible domain, $L$: constant of $L$-smooth for the objective function. \\
    {\bf Output:} the solution $\mathbf{x}(1)$ with probability $1-p$ and the solution $\mathbf{y}$ with probability $p$
    
	\begin{algorithmic}[1]
		\STATE $\mathbf{y} \leftarrow$ Frank-Wolfe with adaptive step sizes $(F,\mathcal{P})$
        \STATE $A$ is a random subset of $[n]$. For all $i\in [n]$, $i \in A$  with probability of $y_i$ independently. \\
        \text{/* \quad Frank-Wolfe Variant Stage \quad */}
        \STATE Initialize $\mathbf{x}(0) = \mathbf{0}$, $t_0=0$, $N_1 = O\left(\frac{\theta DL}{\varepsilon}\right)$, $N_2= O\left(\frac{(1-\theta)DL}{\varepsilon}\right)$.
		\FOR{ $j =0$ to $N_1+N_2-1$}
        \IF{$j \le N_1-1$}
        \STATE $\mathbf{v}(t_j) = \arg\max_{\mathbf{v} \in \mathcal{P}, \mathbf{v} \le (1-\mathbf{x}(t_j))\odot \mathbf{1}_{\hat{A}}} \langle\nabla F(\mathbf{x}(t_j)), \mathbf{v}\rangle$
        \STATE $\mathbf{x}(t_{j+1}) = \mathbf{x}(t_j) + (1-e^{\frac{-1}{N_1+N_2}})  \mathbf{v}(t_j)$
        \ELSE
        \STATE  $\mathbf{v}(t_j) = \arg\max_{\mathbf{v} \in \mathcal{P},\mathbf{v}\le 1-\mathbf{x}(t_j)} \langle\nabla F(\mathbf{x}(t_j)),\mathbf{v} \rangle$
        \STATE $\mathbf{x}(t_{j+1}) = \mathbf{x}(t_j) + \frac{1}{(N_1+N_2)e^{\frac{1}{N_1+N_2}}} \mathbf{v}(t_j)$
        \ENDIF
        \ENDFOR
	\end{algorithmic}
    \label{algorithm 2}
\end{algorithm}

\subsection{Discretization via the Lyapunouv Function}\label{subsection 4.3}
In this subsection, we construct a discrete-time implementation of the above continuous-time algorithm that keeps the same approximation ratio with polynomial-time oracle complexity by following the framework proposed in ~\cite{Du1}. Based on the Lyapunov function in Subsection \ref{subsection 4.1}, we first introduce the discrete version of the Lyapunov function, namely the potential function. From the potential function, the discretized implementation of Algorithm~\ref{algorithm 1} is naturally obtained.

\begin{definition} (\cite{Du-submodular-book2022})
    The potential function associated with a continuous-time algorithm $\mathbf{x}(t)$ is as follows: 
    \begin{equation}\label{potential function}
        E(\mathbf{x}(t_j)) = a(t_j) F(\mathbf{x}(t_j)) + h(t_j) F(\mathbf{y}) - d(t_j) F(\mathbf{x}^*)
    \end{equation}
    And the upper bound of the optimal solution $F(\mathbf{x}^*)$ is as follows:
    \begin{equation}
        F(\mathbf{x}^*) \le p(\mathbf{x}(t_j)) F(\mathbf{x}(t_j)) + q(\mathbf{x}(t_j)) \langle \nabla F(\mathbf{x}(t_j)), v(\mathbf{x}(t_j)) \rangle + r(\mathbf{x}(t_j))F(\mathbf{y})
    \end{equation}
    where $a(t_j)$, $h(t_j)$, $d(t_j)$ are non-negative, non-decreasing sequences and $p(\mathbf{x}(t_j))$, $q(\mathbf{x}(t_j))$, $r(\mathbf{x}(t_j))$ are non-negative sequences.
\end{definition}

Similar as Subsection \ref{subsection 4.1}, if $E(\mathbf{x}(t_j))$ is a non-decreasing sequence, i.e., $E(\mathbf{x}(t_{j+1}))-E(\mathbf{x}(t_j)) \ge 0$, it implies $E(\mathbf{x}(t_N)) - E(\mathbf{x}(t_0)) = \sum_{j=0}^{N-1} E(\mathbf{x}(t_{j+1}))-E(\mathbf{x}(t_j)) \ge 0$, where $N$ is the number of iterations. However, there are some errors in the process of discretization. The first error arises from the continuous algorithm to the discrete algorithm, which is bounded by $L$-smooth, i.e., $F(\mathbf{x}(t_{j+1}))- F(\mathbf{x}(t_j)) \ge \langle \nabla F(\mathbf{x}(t_j)), \mathbf{x}(t_{j+1})-\mathbf{x}(t_j) \rangle -\frac{L}{2} \left\|\mathbf{x}(t_{j+1})-\mathbf{x}(t_j)\right\|^2$. The second error comes from the differential equation by using the finite difference to replace the derivative. Therefore, we introduce some error measures that implies $E(\mathbf{x}(t_{j+1}))-E(\mathbf{x}(t_j)) \ge -Err_1(t_j) F(\mathbf{x}(t_j)) - Err_2(t_j) F(\mathbf{y}) - Err_3(t_j)$. If $Err_1(t_j)$ and $Err_2(t_j)$ are non-positive, it means that we just consider the error comes from the $L$-smooth, i.e., $Err_3(t_j)$. If the solution $\mathbf{x}(t_j)$ in each iteration is feasible and $Err_1(t_j) \ge 0$, we can substitute $F(\mathbf{x}(t_j))$ by $F(\mathbf{x}^*)$. Hence, the approximation ratio has the following closed-form, $Err_1^+(t_j) = \max\{Err_1(t_j),0 \}$, similarly $Err_2^+(t_j)$.
\begin{equation}
    \begin{aligned}
         & \cfrac{a(t_N)}{a(t_N) + h(t_N)-h(t_0)+\sum_{j=0}^{N-1} Err_2^+(t_j)} F(\mathbf{x}(t_N))  + \cfrac{h(t_N)-h(t_0)+\sum_{j=0}^{N-1} Err_2^+(t_j)}{a(t_N) + h(t_N)-h(t_0)+\sum_{j=0}^{N-1} Err_2^+(t_j)} F(\mathbf{y}) \\
         & \quad \ge \cfrac{d(t_N)-d(t_0)-\sum_{j=0}^{N-1} Err_1^+(t_j)}{a(t_N) + h(t_N)-h(t_0)+\sum_{j=0}^{N-1} Err_2^+(t_j)} F(\mathbf{x}^*) -\cfrac{\sum_{j=0}^{N-1} Err_3(t_j)} {a(t_N) + h(t_N)-h(t_0)+\sum_{j=0}^{N-1} Err_2^+(t_j)}
    \end{aligned}
\end{equation}

To achieve our target, plugging the discrete upper bound of $F(\mathbf{x}^*)$ into $E(\mathbf{x}(t_{j+1}))- E(\mathbf{x}(t_j))$, we have the lower bound of $E(\mathbf{x}(t_{j+1}))- E(\mathbf{x}(t_j))$. In light of the lower bound, it induces a sufficient condition to guarantee the non-decreasing $E(\mathbf{x}(t_j))$ with some errors.

\begin{equation}\label{discrete lower bound of E}
\begin{array}{c}
    \begin{aligned}
        E(\mathbf{x}(t_{j+1})) - E(\mathbf{x}(t_j)) & = a(t_{j+1}) [F(\mathbf{x}(t_{j+1}))- F(\mathbf{x}(t_j))] + [a(t_{j+1})- a(t_j)] F(\mathbf{x}(t_j))\\
        & \quad + [h(t_{j+1}) - h(t_j) ] F(\mathbf{y})- [d(t_{j+1})-d(t_j)] F(\mathbf{x}^*) \\
        & \ge \langle \nabla F(\mathbf{x}(t_j)), a(t_{j+1}) [\mathbf{x}(t_{j+1})-\mathbf{x}(t_j)] - [d(t_{j+1})-d(t_j)] q(\mathbf{x}(t_j)) v(\mathbf{x}(t_j))\rangle \\
        & \quad -\frac{a(t_{j+1})L}{2} \left\|\mathbf{x}(t_{j+1})-\mathbf{x}(t)_j\right\|^2\\ 
        & \quad + [a(t_{j+1})- a(t_j) - [d(t_{j+1})-d(t_j)] p(\mathbf{x}(t_j))] F(\mathbf{x}(t_j)) \\
        & \quad + [h(t_{j+1}) - h(t_j) - [d(t_{j+1})-d(t_j)] r(\mathbf{x}(t_j)) ] F(\mathbf{y})\\
        & \ge -Err_1(t_j) F(\mathbf{x}(t_j)) - Err_2(t_j) F(\mathbf{y})- Err_3(t_j)
    \end{aligned} \\
    \bigg \Downarrow \\
   \begin{cases}
        a(t_{j+1})- a(t_j) - [d(t_{j+1})-d(t_j)] p(\mathbf{x}(t_j)) & = -Err_1(t_j) \\
        h(t_{j+1}) - h(t_j) - [d(t_{j+1})-d(t_j)] r(\mathbf{x}(t_j)) & = -Err_2(t_j)\\
        \cfrac{a(t_{j+1})L}{2} \left\|\mathbf{x}(t_{j+1})-\mathbf{x}(t_j)\right\|^2 & = Err_3(t_j) \\
        \cfrac{[d(t_{j+1})-d(t_j)] q(\mathbf{x}(t_j)) v(\mathbf{x}(t_j))}{a(t_{j+1})} & =\mathbf{x}(t_{j+1})-\mathbf{x}(t_j)
    \end{cases}
\end{array}
\end{equation}

Note that $Err_3(t_j)$ is non-negative, which implies that the error caused by the discrete algorithm is unavoidable. If there are some special sequences, which ensure that $Err_1(t_j)$ or $Err_2(t_j)$ is non-positive, we can ignore $F(\mathbf{x}(t_j))$ or $F(\mathbf{y})$, accordingly. Therefore, we construct a large-scale numerical optimization problem to maximize the approximation ratio, whose decision variables are the sequence $a(t_j)$, $h(t_j)$, $d(t_j)$, and $v(\mathbf{x}(t_j))$.

\begin{equation}\label{potential optimization problem}
    \begin{aligned}
        \max \quad & \left \{ \cfrac{d(t_N)-d(t_0)-\sum_{j=0}^{N-1} Err_1^+(t_j)}{a(t_N) + h(t_N)-h(t_0)+\sum_{j=0}^{N-1} Err_2^+(t_j)}, \cfrac{-\sum_{j=0}^{N-1} Err_3(t_j)} {a(t_N) + h(t_N)-h(t_0)+\sum_{j=0}^{N-1} Err_2^+(t_j)} \right \}\\
        \text{s.t.} \quad & a(t_j), b(t_j), c(t_j), d(t_j) \text{ satisfies Eq.(\ref{discrete lower bound of E})}
    \end{aligned}
\end{equation}

In addition to solve the above optimization problem, we can also identify sequence $a(t_j)$, $h(t_j)$, $d(t_j)$ by using parameter functions in the Lyapunov function. It implies that the objective is to determine the sequence of $t_j$. If we choose a proper sequence $t_j$ such that $\sum_{j=0}^{N-1} Err_1(t_j)$, $\sum_{j=0}^{N-1} Err_2(t_j)$, and $\sum_{j=0}^{N-1} \frac{a(t_{j+1})L}{2a_{t_N}} \left\|\mathbf{x}(t_{j+1})-\mathbf{x}(t)_j\right\|^2$ are small enough, the approximation ratio of the discrete implementation has a small loss on the approximation ratio of the Lyapunov function. Based on Tab.\ref{Parameter Function in Lyapunov Function}, we calculate errors that comes from the differential equation.
\begin{equation}
\begin{aligned}
    Err_1(t_j) & = \begin{cases}
        0  &  t_j \le \theta \\
        e^{t_j}(1+t_{j+1}-t_j)-e^{t_{j+1}} \le 0 & t_j > \theta
    \end{cases}\\
    Err_2(t_j) & = \begin{cases}
        1+t_{j+1}-t_j-e^{t_{j+1}-t_j} \le 0 & t_j \le \theta \\
        0 & t_j > \theta
    \end{cases}
\end{aligned}
\end{equation}
where inequalities are due to $ e^{t_{j+1}-t_j}\ge1+t_{j+1}-t_j$. It implies that errors of differential equations are negligible, that is, we can release $F(\mathbf{x}(t_j))$ and $F(\mathbf{y})$ in $E(\mathbf{x}(t_{j+1}))-E(\mathbf{x}(t_j))$. In other words, $E(\mathbf{x}(t_{j+1}))-E(\mathbf{x}(t_j)) \ge -\frac{a(t_{j+1})L}{2} \left\|\mathbf{x}(t_{j+1})-\mathbf{x}(t_j)\right\|^2$. Furthermore, we have $t_N =1 $ and $t_0 = 0$. Therefore, the approximation ratio of the discrete implementation of the continuous algorithm $\mathbf{x}(t)$ is as follows:
\begin{equation}\label{approximation ratio in potential function}
    \begin{aligned}
         & \cfrac{e}{e+e^{\theta}\cdot (3-\theta)-3}F(\mathbf{x}(t_N)) + \cfrac{e^{\theta}\cdot (3-\theta)-3}{e+e^{\theta}\cdot (3-\theta)-3} F(\mathbf{y})\\
         & \quad \ge \cfrac{e^\theta\cdot (2-\theta)-1}{e+e^{\theta}\cdot (3-\theta)-3} F(\mathbf{x}^*) - \sum_{j=0}^{N-1} \cfrac{a(t_{j+1})L}{2 (e+e^{\theta}\cdot (3-\theta)-3)}\left\|\mathbf{x}(t_{j+1})-\mathbf{x}(t_j)\right\|^2 
    \end{aligned}
\end{equation}

Note that if we can restrict the last term small enough in the above inequality, it implies that the discrete implementation is a $0.385$ approximation algorithm to maximize general continuous DR-submodular functions subject to down-closed convex constraints. Additionally, the last term is often called convergence rate. Next, we analyze the complexity of the discrete implementation of the Algorithm~\ref{algorithm 1}. For convenience, let $N_1$ be the number of iterations in $t \in [0,\theta]$, $N_2$ be the number of iterations in $t \in (\theta,1]$, $D$ be the diameter of the feasible domain, i.e., $D= \max_{\mathbf{x},\mathbf{y} \in \mathcal{P}} \left\|\mathbf{x}-\mathbf{y}\right \|^2$, and $t_{j+1}-t_j = \frac{1}{N_1+N_2}=\frac{1}{N}$.

\begin{equation}
    \begin{aligned}
       & \sum_{j=0}^{N-1} \cfrac{a(t_{j+1})L}{2 (e+e^{\theta}\cdot (3-\theta)-3)} \left\|\mathbf{x}(t_{j+1})-\mathbf{x}(t_j)\right\|^2 = \sum_{j=0}^{N-1} \cfrac{[d(t_{j+1})-d(t_j)]^2 q(t_j)^2L}{2a(t_{t+1}) (e+e^{\theta}\cdot (3-\theta)-3)} \left\|v(\mathbf{x}(t_j)\right\|^2\\
        & \le \sum_{j=0}^{N_1-1} \cfrac{[e^{t_{j+1}}-e^{t_j}]^2 LD}{2e^{t_{j+1}}(e+e^{\theta}\cdot (3-\theta)-3)} + \sum_{j=N_1-1}^{N-1} \cfrac{e^{2t_j}(t_{j+1}-t_j)^2 LD}{2e^{t_{j+1}}(e+e^{\theta}\cdot (3-\theta)-3)} \\
        & = \cfrac{DL}{2 (e+e^{\theta}\cdot (3-\theta)-3)}\cdot \left[ (1-e^{\frac{-1}{N_1+N_2}}) \cdot (e^{\theta}-1) +\cfrac{e^{\theta}-e}{(N_1+N_2)^2 e^{\frac{1}{N_1+N_2}} (1-e^{\frac{1}{N_1+N_2}}) } \right] \\
        & \le \cfrac{DL}{2 (e+e^{\theta}\cdot (3-\theta)-3)} \cdot \left[\cfrac{e^\theta-1}{N_1+N_2} + \cfrac{e-e^{\theta}}{(N_1+N_2)^2 \frac{1+N_1+N_2}{N_1+N_2}\frac{1}{N_1+N_2}   } \right] = O\left(\cfrac{DL}{N_1+N_2}\right) = \varepsilon
    \end{aligned}  
\end{equation}
The first inequality comes from the definition of $D$, i.e., $\left\|v(\mathbf{x}(t_j))\right\|^2 \le D$. The second equality is due to $e^{t_0} = 1$, $e^{t_{N_1}} = e^{\theta}$, $e^{t_{N}} = e$, $t_{j+1}-t_j = \frac{1}{N_1+N_2}$ and $1+x\le e^{x}$. Therefore, we have $N_1+N_2 = O\left(\frac{DL}{\varepsilon}\right)$. In addition, we have $\frac{N_1}{N_1+N_2} = \theta$. It implies $N_1 = O\left(\frac{\theta DL}{\varepsilon}\right)$ and $N_2 = O\left(\frac{(1-\theta)DL}{\varepsilon}\right)$. Therefore, we have the following conclusion.

\begin{theorem}
    Algorithm~\ref{algorithm 2} outputs a solution that is at least $0.385$ times of the optimal value for problem (\ref{eq:main_prob}) with a gap smaller than $\varepsilon$ with at most $O\left(\frac{DL}{\varepsilon}\right)$ iterations, where $D$ is the diameter of the feasible domain and $L$ is the L-smooth constant of the DR-submodular function $F$.
\end{theorem}

Note that a rephrased version of Lemma 2.3.7 in \cite{Feldman2} implies that $O(DL) = O\left(n^3\right) F(\mathbf{x}^*)$ for submodular set functions. Therefore, Algorithm~\ref{algorithm 2}  is a $0.385-O(n^{-1})$ approximation algorithm for problem (\ref{eq:main_prob}) with $O\left(\frac{DL}{\varepsilon}\right) = O\left(n^4\right)$ iterations, i.e., $N_1= O\left(\theta n^4\right)$ and $N_2=O\left((1-\theta) n^4\right)$. In addition, it is easy to verify that approximating $e^{t_j}$ with $\left(1-\frac{1}{N_1+N_2}\right)^{-j}$ leads to step size $\frac{1}{N_1+N_2}$ and $\left(1-\frac{1}{N_1+N_2}\right)\frac{1}{N_1+N_2}$, which mathces the results in \cite{Buchbinder2}. 

\subsection{Discussions}
In general, the partition $(A,\hat{A})$ in Lemma \ref{lemma 14} is an adaptive partition, which can change over time $t$ subject to $\theta_A(t) \le \theta_{\hat{A}}(t)$. Especially, the partition $(A,\hat{A})$ can also be pre-determined before starting the continuous algorithm $\mathbf{x}(t)$. Therefore, we should guarantee the following inequality over each time $t$ for the pre-determined partition.
\begin{equation*}
    \int_0^t \dot{\mathbf{x}}_A(t) dt \le \int_0^t \dot{\mathbf{x}}_{\hat{A}}(t) dt
\end{equation*}

Although the above analysis proves the optimality of the algorithm to some extent under a partition that is independent of time, whether there are suitable time-dependent partitions to lead to a better approximation is still an open question. 

Claim \ref{claim} illustrates that the worst stationary point without any approximation ratio is probably close to the boundary of the feasible domain in some dimensions. This phenomenon inspires a natural idea to obtain a potentially improved solution $\mathbf{x}(t)$, i.e., $\theta_A(t) \le \theta_{\hat{A}}(t)$, where $A$ is a dimension set which are the dimensions near the boundary of the feasible domain in $\mathbf{y}$. Therefore, we choose the pre-determined partition $A, \hat{A}$ depending on the stationary point $\mathbf{y}$. However, there is another important remaining question, whether there are better choices for $\mathbf{y}$ except stationary points. If yes, it potentially improves the $0.385$ approximation algorithm for the problem (\ref{eq:main_prob}).

\section{Numerical Experiments}
In this section, we test five different algorithms on some applications mentioned in Section \ref{section 1}. Two of them serve as baseline algorithms to contrast algorithms in our work.
\begin{itemize}
    \item Baseline algorithms: \textsc{Two-Phase Frank-Wolfe} and \textsc{Frank-Wolfe Variant} have been proposed by \cite{Bian1} to maximize continuous DR-submodular functions subject to down-closed convex constraints. The first one is a $1/4$ approximation algorithm, which finds a stationary point $\mathbf{x}$ in the feasible domain $\mathcal{P}$ and a stationary point $\mathbf{y}$ in the feasible domain $\mathcal{P}\wedge (\mathbf{u}-\mathbf{x})$ and returns the best of them. The second one is an $1/e$ approximation algorithm, which returns a solution $\mathbf{x}(1)$ in feasible domain $\mathcal{P}$ where $x_i(1)-x_i(0) \le (1-e^{-1}) (u_i-l_i)$ for all $i \in N$.  
    \item Algorithms in this work: we adopt \textsc{Frank-Wolfe} to return a stationary point, which is proved in Subsection \ref{subsection 3.1} by a concrete example that the stationary point may be arbitrarily bad. In addition, we adopt \textsc{Frank-Wolfe} to return a stationary point in the reduce feasible domain, namely \textsc{Reduced Frank-Wolfe}, which has the $0.309$ approximation guarantee. \textsc{Aided Frank-Wolfe Variant} is a $0.385$ approximation algorithm analyzed in Section \ref{section 4}. Although it returns the stationary point $\mathbf{y}$ with probability $0.23$ and the solution $\mathbf{x}(1)$ returned by \textsc{Frank-Wolfe Variant} stage with probability $0.77$, we make a slight change, which returns the best one in $\mathbf{x}(1)$ and $\mathbf{y}$. It does not cause any loss to the approximation ratio.   
\end{itemize}
Because both \textsc{Two-Phase Frank-Wolfe} and \textsc{Aided Frank-Wolfe Variant} need to find a stationary point in the feasible domain $\mathcal{P}$, in the following numerical experiments, we adopt the stationary point returned by \textsc{Frank-Wolfe} as their stationary points. Additionally, in the following tables and figures, without causing ambiguity, \textsc{Two-Phase Frank-Wolfe-2} represents the second stationary point returned by \textsc{Two-Phase Frank-Wolfe}, and \textsc{Aided Frank-Wolfe Variant-2} represents the solution returned by \textsc{Frank-Wolfe variant} stage in \textsc{Aided Frank-Wolfe Variant}. In other words, the solution returned by \textsc{Two-Phase Frank-Wolfe} is the best one of \textsc{Frank-Wolfe} and \textsc{Two-Phase Frank-Wolfe-2}; the solution returned by \textsc{Aided Frank-Wolfe Variant} is the best one of the \textsc{Frank-Wolfe} and \textsc{Aided Frank-Wolfe Variant-2}. In our experiments, the number of iterations is $N=100$. Because there is randomness in all examples, the number of repeated trails is $T=10$. 

\begin{table}[htb]
    \centering
    \caption{Average function value with standard deviation in NQP}
    \begin{tabular}{m{1.2cm} m{3.3cm} m{3.1cm} m{3.1cm} m{3.1cm}}
        & Algorithm & $m=\lfloor 0.1n \rfloor$ & $m=\lfloor 0.3n \rfloor$ & $m=\lfloor 0.5n \rfloor$ \\
         \hline
       \multirow{9}{*}{$n=100$} & \textsc{Frank-Wolfe} & $843.2299 \pm 5.4480$ & $836.0062 \pm 3.3833$ & $759.4658 \pm 9.2937$\\
       \cline{2-5}
       & \textsc{Two-Phase Frank-Wolfe-2} & $837.0060 \pm 4.6215$ & $833.4710 \pm 3.5045$ & $759.2584 \pm 8.7699$\\
       \cline{2-5}
       & \textsc{Reduced Frank-Wolfe}& $836.5517 \pm 5.2064$ & $833.1673 \pm 3.4484$ & $759.3279 \pm 8.9670$\\
       \cline{2-5}
       & \textsc{Frank-Wolfe Variant}& $840.1015 \pm 5.1331$ & $834.7916 \pm 3.6254$ & $759.4359 \pm 9.2176$\\
       \cline{2-5}
       & \textsc{Aided Frank-Wolfe Variant-2}& $838.0243 \pm 4.9410$ & $833.6677 \pm 3.6423$ & $759.1773 \pm 9.0697$\\
       \hline
       \multirow{9}{*}{$n=500$} & \textsc{Frank-Wolfe} & $20806.15 \pm 20.6679$ & $20699.98 \pm 23.9101$ & $18818.92 \pm 35.9268$\\
       \cline{2-5}
       & \textsc{Two-Phase Frank-Wolfe-2} & $20759.12 \pm 20.2773$ & $20675.61 \pm 21.4277$ & $18815.96 \pm 33.5124$ \\
       \cline{2-5}
       & \textsc{Reduced Frank-Wolfe}& $20750.54 \pm 19.6014 $ & $20674.48 \pm 21.4754$ & $18816.51 \pm 34.1660$\\
       \cline{2-5}
       & \textsc{Frank-Wolfe Variant}& $20783.49 \pm 18.9631$ & $20689.62 \pm 21.6035$ & $18817.98 \pm 34.9901$\\
       \cline{2-5}
       & \textsc{Aided Frank-Wolfe Variant-2}& $20760.85 \pm 19.2939$ & $20672.05 \pm 21.0048$ & $18815.81 \pm 33.9950$\\
       \hline 
       \multirow{9}{*}{$n=1000$} & \textsc{Frank-Wolfe} & $82979.21 \pm 60.7779$ & $82762.57 \pm 69.0859$ & $75125.82 \pm 65.5320$\\
       \cline{2-5}
       & \textsc{Two-Phase Frank-Wolfe-2} & $82857.82 \pm 72.6572$ & $82694.05 \pm 63.6619$ & $75125.25 \pm 64.8827$ \\
       \cline{2-5}
       & \textsc{Reduced Frank-Wolfe}& $82833.04 \pm 67.5125$ & $82690.23 \pm 64.5268$ & $75125.12 \pm 64.8164$\\
       \cline{2-5}
       & \textsc{Frank-Wolfe Variant}& $82921.56 \pm 69.3485$ & $82736.38 \pm 66.0542$ & $75125.74 \pm 65.4271$\\
       \cline{2-5}
       & \textsc{Aided Frank-Wolfe Variant-2}& $82850.28 \pm 69.7793$ & $82678.13 \pm 62.3230$ & $75123.11 \pm 64.1905$\\
       \hline 
    \end{tabular}
    \label{NQP}
\end{table}

\subsection{Non-monotone Quadratic Programming}
Our first instance is the synthetic function of non-monotone quadratic programming. The objective function is a non-monotone DR-submodular function $F(\mathbf{x})= \frac{1}{2} \mathbf{x}^T \mathbf{H} \mathbf{x} + \mathbf{h}^T \mathbf{x} + c$ and the constraints are linear inequalities $\mathcal{P} = \{\mathbf{x} \in [\mathbf{0},\mathbf{1}]| \mathbf{A}\mathbf{x} \le \mathbf{b}, \mathbf{A} \in \mathbb{R}^{m\times n}_{+}, \mathbf{b} \in \mathbb{R}^m_+ \}$. We randomly generate the matrix $\mathbf{H}$, whose entries are uniformly distributed in $[-1,0]$, and the vector $\mathbf{h}$, whose components are uniformly distributed in $[0,1]$. To ensure that $F(\mathbf{x})$ is non-negative, we choose $c$ to be a large enough number. As for constraints, we randomly generate the matrix $\mathbf{A}$, whose entries are uniformly distributed in $[0,1]$, and the vector $\mathbf{b}$, whose components are uniformly distributed in $[0,0.05n]$, where we set $m=\lfloor 0.1n \rfloor, \lfloor 0.3n \rfloor, \lfloor 0.5n \rfloor$, respectively; and $n=100, 500, 1000$, respectively.  

Tab.\ref{NQP} shows the averages and standard deviations of the objective value over $10$ repeated trails. Although the stationary point is probably arbitrarily bad in the worst cases, it is the best solution in our numerical experiments. In other words, \textsc{Two-Phase Frank-Wolfe} and \textsc{Aided Frank-Wolfe Variant} have the best performance in \textsc{Two-Phase Frank-Wolfe}, \textsc{Aided Frank-Wolfe Variant}, \textsc{Reduced Frank-Wolfe}, and \textsc{Frank-Wolfe Variant}. This is consistent with the results in \cite{Bian1}, who illustrate that the \textsc{Two-Phase Frank-Wolfe} performs better than \textsc{Frank-Wolfe Variant}. However, our results more clearly show that \textsc{Two-Phase Frank-Wolfe} has good performances because the stationary point is good enough. In other words, \textsc{Two-Phase Frank-Wolfe-2} and \textsc{Aided Frank-Wolfe Variant-2} do not improve the quality of our solution in our experiments. In addition, \textsc{Reduced Frank-Wolfe} may only have theoretical value. Despite the approximation ratio of \textsc{Reduced Frank-Wolfe} being larger than that of \textsc{Two-Phase Frank-Wolfe}, the actual performance of the latter is better than that of the former.

\begin{table}[htb]
    \centering
    \caption{Average function value with standard deviation in regular coverage}
    \begin{tabular}{m{1.2cm} m{3.3cm} m{3.5cm} m{3.5cm} m{3.5cm}}
        & Algorithm & $m=\lfloor 0.1(2k+1) \rfloor$ & $m=\lfloor 0.3(2k+1) \rfloor$ & $m=\lfloor 0.5(2k+1) \rfloor$ \\
         \hline
       \multirow{9}{*}{$k=10$} & \textsc{Frank-Wolfe} & $ 9.9997\pm 0.0001$ & $9.9998 \pm 6.15e-05$ & $9.9999 \pm 4.13e-05$\\
       \cline{2-5}
       & \textsc{Two-Phase Frank-Wolfe-2} & $1.0000 \pm 4.68e-09$ & $0.9999 \pm 1.41e-05$ & $0.9895 \pm 0.0285$\\
       \cline{2-5}
       & \textsc{Reduce Frank-Wolfe} & $3.8851 \pm 0.0037$ & $3.8945 \pm 0.0095$ & $3.9011 \pm 0.0087$\\
       \cline{2-5}
       & \textsc{Frank-Wolfe Variant}& $6.0734 \pm 0.0518$ & $6.1365 \pm 0.0291$ & $6.1546 \pm 0.0214$\\
       \cline{2-5}
       & \textsc{Aided Frank-Wolfe Variant-2}& $4.3887 \pm 0.0559$ & $4.4674 \pm 0.0259$ & $4.4884 \pm 0.0178$\\
       \hline
       \multirow{9}{*}{$k=100$} & \textsc{Frank-Wolfe} & $99.9957 \pm 7.64e-05$ & $99.9961 \pm 4.09e-05$ & $99.9963 \pm 5.22e-05$\\
       \cline{2-5}
       & \textsc{Two-Phase Frank-Wolfe-2} & $1.0000 \pm 2.15e-14$ & $1.0000 \pm 2.43e-14$ & $1.0000 \pm 1.21e-14$ \\
       \cline{2-5}
       & \textsc{Reduce Frank-Wolfe}& $38.2350 \pm 0.0043 $ & $38.2389 \pm 0.0006$ & $38.2431 \pm 0.0021$\\
       \cline{2-5}
       & \textsc{Frank-Wolfe Variant}& $62.2366 \pm 0.0257$ & $62.3166 \pm 0.0134$ & $62.2949 \pm 0.0151$\\
       \cline{2-5}
       & \textsc{Aided Frank-Wolfe Variant-2}& $42.5358 \pm 0.0565$ & $42.8831 \pm 0.0324$ & $43.0078 \pm 0.0380$\\
       \hline 
       \multirow{9}{*}{$k=500$} & \textsc{Frank-Wolfe} & $499.9791 \pm 8.35e-05$ & $499.9794 \pm 6.37e-05$ & $499.9803 \pm 7.86e-05$\\
       \cline{2-5}
       & \textsc{Two-Phase Frank-Wolfe-2} & $1.0000 \pm 6.14e-14$ & $1.0000 \pm 7.1e-14$ & $1.0000 \pm 4.63e-14$ \\
       \cline{2-5}
       & \textsc{Reduced Frank-Wolfe}& $190.5986 \pm 0.0068$ & $190.6434 \pm 0.0049$ & $190.6910 \pm 0.0087$\\
       \cline{2-5}
       & \textsc{Frank-Wolfe Variant}& $314.6732 \pm 0.0105$ & $314.7026 \pm 0.0093$ & $314.7211 \pm 0.0106$\\
       \cline{2-5}
       & \textsc{Aided Frank-Wolfe Variant-2}& $218.0404 \pm 0.0510$ & $218.1358 \pm 0.0237$ & $218.1503 \pm 0.0346$\\
       \hline 
    \end{tabular}
    \label{regular coverage}
\end{table}

\subsection{Regular Coverage Maximization}
Our second instance is the multilinear extension of the regular coverage function in Example \ref{example 1}, where the constraints are linear inequalities $\mathcal{P} = \{\mathbf{x} \in [\mathbf{0},\mathbf{1}]| \mathbf{A}\mathbf{x} \le \mathbf{b}, \mathbf{A} \in \mathbb{R}^{m\times (2k+1)}_{+}, \mathbf{b} \in \mathbb{R}^{2k+1}_+ \}$, $k=10, 100, 500$, and $m=\lfloor 0.1(2k+1) \rfloor, \lfloor 0.3(2k+1) \rfloor, \lfloor 0.5(2k+1) \rfloor$, respectively. Each entry in the matrix $\mathbf{A}$ is uniformly distributed in $[0,1]$. Each component in the vector $\mathbf{b}$ is uniformly distributed in $[0,0.05(2k+1)]$.

Tab.\ref{regular coverage} shows that the stationary point returned by \textsc{Frank-Wolfe} is the best among all algorithms for this problem; that is, \textsc{Two-Phase Frank-Wolfe} and \textsc{Aided Frank-Wolfe Variant} have the best performance. In addition, Tab.\ref{regular coverage} is consist with our theoretical analyses in Subsection \ref{subsection 3.1}. The \textsc{Two-Phase Frank-Wolfe-2} row in Tab.\ref{regular coverage} illustrates that \textsc{Frank-Wolfe} can indeed return a poor solution in numerical experiments, which implies that \textsc{Aided Frank-Wolfe Variant} potentially returns a better solution than \textsc{Two-Phase Frank-Wolfe} for this problems. Similar to NQP, \textsc{Reduced Frank-Wolfe} seems to have only theoretical values.

\begin{table}[htb]
    \centering
    \caption{Average function value with standard deviation in DPPs}
    \begin{tabular}{m{1.2cm} m{3.3cm} m{3.1cm} m{3.1cm} m{3.1cm}}
        & Algorithm & $m=\lfloor 0.1n \rfloor$ & $m=\lfloor 0.3n \rfloor$ & $m=\lfloor 0.5n \rfloor$ \\
         \hline
       \multirow{9}{*}{$n=10$} & \textsc{Frank-Wolfe} & $ 0.8767\pm 0.5306$ & $0.6857 \pm 0.3032$ & $0.4604 \pm 0.1822$\\
       \cline{2-5}
       & \textsc{Two-Phase Frank-Wolfe-2} & $0.4678 \pm 0.2265$ & $0.5341 \pm 0.1589$ & $0.4499 \pm 0.1799$\\
       \cline{2-5}
       & \textsc{Reduced Frank-Wolfe}& $0.6342 \pm 0.3337$ & $0.6033 \pm 0.2212$ & $0.4468 \pm 0.1809$\\
       \cline{2-5}
       & \textsc{Frank-Wolfe Variant}& $0.7709 \pm 0.4389$ & $0.6473 \pm 0.2624$ & $0.4581 \pm 0.1830$\\
       \cline{2-5}
       & \textsc{Aided Frank-Wolfe Variant-2}& $0.6795 \pm 0.3778$ & $0.6005 \pm 0.2302$ & $0.4327 \pm 0.1754$\\
       \hline
       \multirow{9}{*}{$n=100$} & \textsc{Frank-Wolfe} & $1.2722 \pm 0.4443$ & $4.2817 \pm 0.6281$ & $4.4901 \pm 0.5403$\\
       \cline{2-5}
       & \textsc{Two-Phase Frank-Wolfe-2} & $1.0346 \pm 0.3638$ & $4.1647 \pm 0.6557$ & $4.4551 \pm 0.5859$ \\
       \cline{2-5}
       & \textsc{Reduced Frank-Wolfe}& $1.1455 \pm0.3779 $ & $4.2070 \pm 0.6327$ & $4.4636 \pm 0.5658$\\
       \cline{2-5}
       & \textsc{Frank-Wolfe Variant}& $1.2216 \pm 0.4090$ & $4.2695 \pm 0.6336$ & $4.4838 \pm 0.5475$\\
       \cline{2-5}
       & \textsc{Aided Frank-Wolfe Variant-2}& $1.1309 \pm 0.3730$ & $4.0572 \pm 0.6179$ & $4.3070 \pm 0.5600$\\
       \hline 
       \multirow{9}{*}{$n=500$} & \textsc{Frank-Wolfe} & $2.3437 \pm 0.8585$ & $20.5876 \pm 1.1112$ & $19.4039 \pm 1.0967$\\
       \cline{2-5}
       & \textsc{Two-Phase Frank-Wolfe-2} & $2.1179 \pm 0.8734$ & $20.5734 \pm 1.1179$ & $19.3901 \pm 1.1071$ \\
       \cline{2-5}
       & \textsc{Reduced Frank-Wolfe}& $2.2202 \pm 0.8642$ & $20.5634 \pm 1.1218$ & $19.3834 \pm 1.1049$\\
       \cline{2-5}
       & \textsc{Frank-Wolfe Variant}& $2.2986 \pm 0.8698$ & $20.5818 \pm 1.1126$ & $19.3961 \pm 1.0962$\\
       \cline{2-5}
       & \textsc{Aided Frank-Wolfe Variant-2}& $2.1657 \pm 0.8352$ & $19.8884 \pm 1.1097$ & $18.7706 \pm 1.0875$\\
       \hline 
    \end{tabular}
    \label{DPPs}
\end{table}

\subsection{Determinantal Point Processes}
The third instance is the softmax extension of determinantal point processes. We randomly generate a diagonal matrix $\mathbf{D}$ whose $n$ eigenvalues are uniformly distributed in $[-0.5,1.0]$ and take the power of $e$. Then, we construct a positive semidefinite matrix by $\mathbf{L} = \mathbf{V} \mathbf{D} \mathbf{V}^T$ where $\mathbf{V}$ is a random unitary matrix. Similarly, the constraints are linear inequalities $\mathcal{P} = \{\mathbf{x} \in [\mathbf{0},\mathbf{1}]| \mathbf{A}\mathbf{x} \le \mathbf{b}, \mathbf{A} \in \mathbb{R}^{m\times n}_{+}, \mathbf{b} \in \mathbb{R}^{n}_+ \}$ where entries of the matrix $\mathbf{A}$ are uniformly distributed in $[0,1]$, components of the vector $\mathbf{b}$ are uniformly distributed in $[0,0.05n]$, $m=\lfloor 0.1n \rfloor, \lfloor 0.3n \rfloor, \lfloor 0.5n \rfloor$, and $n=10, n=100, n=500$, respectively.  

In Tab.\ref{DPPs}, the results also imply that the stationary point is a good enough solution, that is, \textsc{Two-Phase Frank-Wolfe} and \textsc{Aided Frank-Wolfe Variant} perform better than \textsc{Reduced Frank-Wolfe} and \textsc{Frank-Wolfe Variant}. The solution returned by \textsc{Aided Frank-Wolfe Variant-2} is better than the solution returned by \textsc{Two-Phase Frank-Wolfe-2} in small-scale problems. Similarly, there is the same relationship between \textsc{Aided Frank-Wolfe Variant-2} and \textsc{Reduced Frank-Wolfe}. This phenomenon may be because the number of stationary points increases with the scale of the problem but the number of poor stationary points increases slowly with the scale of the problem. In other words, there is less chance of obtaining a bad stationary point in large-scale problems. 

\subsection{Revenue/Profit Maximization}
The last instance is the revenue/profit maximization problem in online social networks. We perform our algorithms in three networks. One is the synthetic network, which has $2708$ nodes and $5208$ edges. The remaining two networks come from the Stanford Large Network Dataset Collection (\cite{snapnets}). The dataset of Facebook has $4039$ nodes and $88234$ edges and the dataset of Twitter has $81306$ nodes and $1768149$ edges. We set the weight of each edge $(u,v)$ uniformly distributed from $\{0.1,0.01,0.001\}$. Specially, in Eq.(\ref{revenue function}), we adopt the same parameter functions as \cite{Bian}, i.e., $R_s(\mathbf{x}) = \sqrt{\sum_{t:x_t \neq 0} x_t w_{st}}$, $\phi(x_t) = w_{tt} x_t$, $\bar{R}_t(\mathbf{x}) = -x_t$, where $w_{st}$ is the weight of the edge $(s,t)$, $\alpha=\beta=10$, and $\gamma= 0.5$. Because the edge from node $u$ to itself is not included in the three networks, we set $w_{tt} = 0.1$ for each node $t$. In addition, we test our algorithms under two constraints. One is the cardinality constraint, which is commonly known as the profit maximization problem, i.e., $\sum_{i} x_i \le b$. The other one is a system of linear inequalities $\mathcal{P} = \{\mathbf{x} \in [\mathbf{0},\mathbf{1}]| \mathbf{A}\mathbf{x} \le \mathbf{b}, \mathbf{A} \in \mathbb{R}^{m\times n}_{+}, \mathbf{b} \in \mathbb{R}^{n}_+ \}$ where entries of the matrix $\mathbf{A}$ are uniformly distributed in $[0,1]$ and $\mathbf{b}$ is $0.1n$, $0.3n$, and $0.5n$ respectively.

Fig.\ref{cardinality constraint} is the average profit value under cardinality constraints. Fig.\ref{linear system} is the average revenue value under a system of linear inequalities. They show that the stationary point $\mathbf{x}$ in $\mathcal{P}$ and the stationary point $\mathbf{y}$ in $\mathcal{P} \wedge (\mathbf{1}-\mathbf{x})$ perform better than other solutions. This means \textsc{Two-Phase Frank-Wolfe} and \textsc{Aided Frank-Wolfe Variant} have better practical performances than other algorithms. Because the revenue/profit maximization problem in online social networks is usually a large-scale problem, which has thousands or even millions of nodes, it supports our conjecture that there is less chance for algorithms to return a bad stationary point. In addition, Fig.\ref{cardinality constraint} and Fig.\ref{linear system} also illustrate that \textsc{Reduced Frank-Wolfe} is not suitable for practical problems and hence only interesting from a theoretical perspective.

\begin{figure}[htb]
    \begin{minipage}{0.31\linewidth}
        \includegraphics[width=5.2cm]{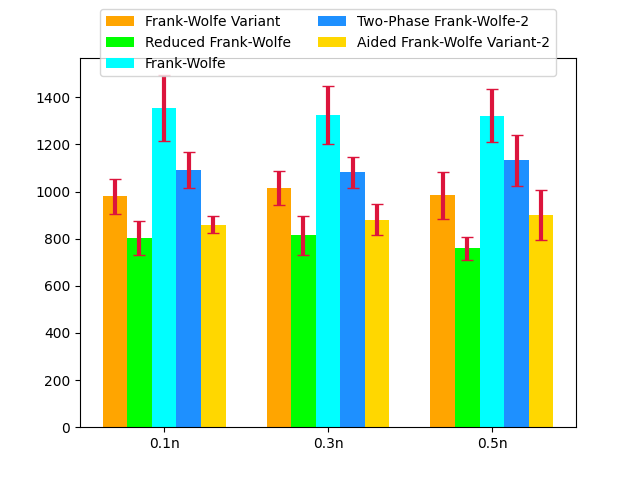}
    \end{minipage}
    \begin{minipage}{0.31\linewidth}
        \includegraphics[width=5.2cm]{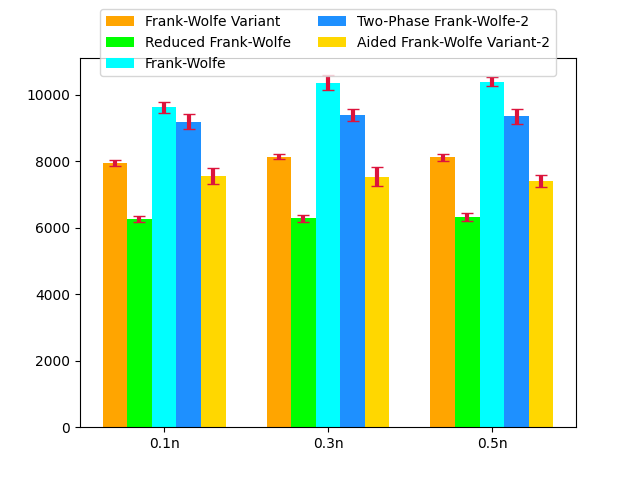}
    \end{minipage}
    \begin{minipage}{0.31\linewidth}
        \includegraphics[width=5.2cm]{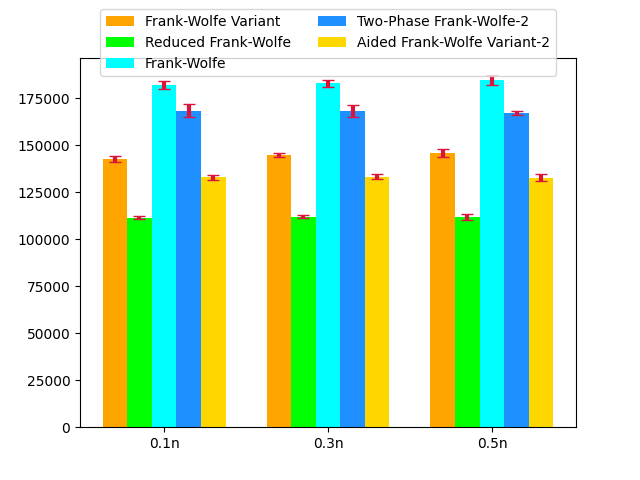}
    \end{minipage}
    \caption{Average function value with standard deviation in revenue maximization under cardinality constraints. From left to right are Synthesis, Facebook, and Twitter.}
    \label{cardinality constraint}
\end{figure}

\begin{figure}[htb]
    \centering
    \begin{minipage}{0.45\linewidth}
        \includegraphics[width=7cm]{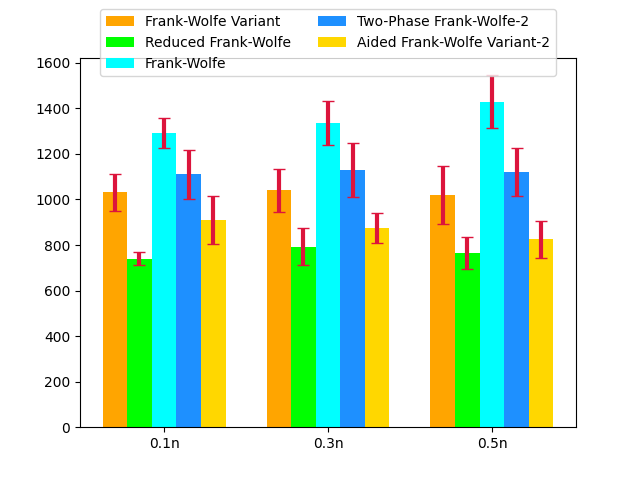}
    \end{minipage}
    \begin{minipage}{0.45\linewidth}
        \includegraphics[width=7cm]{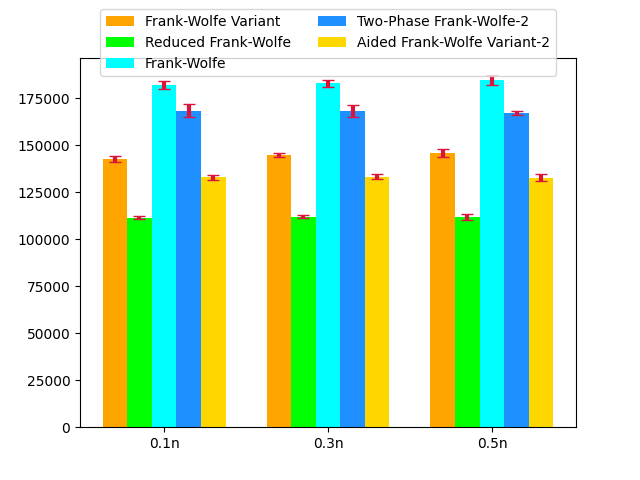}
    \end{minipage}
    \caption{Average function value with standard deviation in revenue maximization under a system of linear inequalities. From left to right are Synthesis and Facebook.}
    \label{linear system}
\end{figure}

\section{Conclusion}
We study the continuous non-monotone DR-submodular maximization problem subject to down-closed convex constraints. Firstly, we illustrate that stationary points can have an arbitrarily bad approximation ratio, a sharp contrast to the monotone DR-submodular case, which always admits a constant approximation~\cite{Hassani}. In particular, we extend the 0.309-approximation algorithm (\cite{Chekuri}) and the $0.385$-approximation algorithm (\cite{Buchbinder2}) to the continuous domain via new analyses without using any properties of the multilinear extension. The most important open question is to close the approximation gap for this problem, where the best lower bound is $0.385$ in continuous domains as shown in this paper and the best upper bound is $0.478$ as shown in \cite{qi2022maximizing}. Our numerical experiments show that both the $0.385$ approximation algorithm and \textsc{two-phase Frank-Wolfe} have excellent practical performances. However, the latter is only shown to be at least $1/4$ (\cite{Bian}). An interesting open question is whether the latter has a better theoretical performance.

\acks{Donglei Du's research is partially supported by NSERC grant (No.283106) and NSFC grants (Nos. 11771386 and 11728104). Wenguo Yang's research is partially supported by NSFC grants (Nos. 12071459 and 11991022). Dachuan Xu's research is partially supported by NSFC grant (No. 12131003).  }











\vskip 0.2in
\bibliography{ref}

\end{document}